\documentclass[11pt]{article}
\usepackage{rheaj}
\usepackage{setspace}
\usepackage{xparse}

\newcommand{\mypara}[1]{\medskip \noindent {\bf #1}}

\newcommand{\hDST}{\textnormal{h-DST}}
\newcommand{\hfracDST}{\textnormal{h-frac-DST}}

\newcommand{\nt}{k}
\newcommand{\lc}{h}

\newcommand{\algsep}{\textsc{PruneAndSeparate}}

\begin{document}
\title{Length-Constrained Network Design in Planar Digraphs\thanks{Siebel School of Computing and Data
    Science, University of Illinois, Urbana-Champaign, Urbana, IL
    61801. {\tt \{chekuri,rheaj3\}@illinois.edu}. Supported in part by
    NSF grant CCF-2402667.}}
\author{Chandra Chekuri \and Rhea Jain}
\date{}
\maketitle
\thispagestyle{empty}

\begin{abstract}
  We study \emph{length-constrained} generalizations of Directed Steiner Tree
  (DST) and Directed Steiner Forest (DSF) in planar digraphs. 
  In both problems, the input is a directed graph with edge costs. 
  DST asks for a min-cost subgraph connecting a root to a given set of
  terminals, and DSF asks for a min-cost subgraph connecting each of
  a given set of source-sink terminal pairs.
  In the length-constrained setting, each edge has both a cost and a
  length, and the input includes a length bound $h$; the goal is to
  find a min-cost subgraph connecting each terminal pair via a path
  of length at most $h$.
  Our work is motivated by a recent line of results showing that several
  network design problems that are traditionally hard in directed
  graphs admit polylogarithmic approximation ratios in planar digraphs.
  We give polylogarithmic bicriteria approximation algorithms for
  length-constrained analogues of DST and DSF in planar digraphs.
  Our approximation ratios match the best known for DST and DSF in
  planar digraphs, with an $O(\log k)$ violation of
  the length constraint, where $k$ denotes the number of
  terminals (or terminal pairs).
  As corollaries, we obtain polylogarithmic approximations for
  \emph{buy-at-bulk} DST and DSF in planar digraphs.
\end{abstract}

\section{Introduction}
\label{sec:intro}

\emph{Directed Steiner Tree} (DST) and \emph{Directed Steiner Forest} (DSF) are fundamental 
problems in network design with a wide array of applications. For both 
problems, the input consists of a directed graph $G = (V, E)$ with edge costs $c: E \to \R_+$. 
In Directed Steiner Tree, given a root $s \in V$ and terminals 
$S \subseteq V \setminus \{s\}$, the goal is to find a min-cost subgraph containing an 
$s$-$t$ path for each $t \in S$. Directed Steiner Forest generalizes DST to the 
\emph{multicommodity} setting: given a set of demand pairs $D = \{(s_i, t_i)\}_{i \in [\nt]}$, 
the goal is to find a min-cost subgraph containing an $s_i$-$t_i$ path for 
each $i \in [\nt]$. 

DST and DSF are $\NP$-hard, and in fact have strong lower bounds 
on their approximability. DST generalizes several hard problems in combinatorial 
optimization, such as Set Cover and Group Steiner Tree. 
It is hard to approximate to an 
$\Omega(\log^2 \nt/\log \log \nt)$-factor under plausible complexity assumptions 
\cite{GrandoniLL22}, and to a slightly weaker $\Omega(\log^{2-\eps} \nt)$-factor 
unless $\NP$ is contained in randomized quasi-poly time \cite{HalperinK03};
here, $\nt$ denotes $|S|$.
The best known approximation ratios are $O(\log^2 \nt/\log \log \nt)$ in 
\emph{quasi-polynomial} time \cite{GrandoniLL22,GhugeN22,Charikaretal99} and 
$O(\nt^{\eps})$ for any fixed $\eps > 0$ in polynomial time \cite{Zelikovsky97}.
Whether DST admits a polylogarithmic approximation in polynomial time is a major 
open question. We remark that the natural LP relaxation for DST has polynomial-factor 
integrality gap \cite{ZosinK02,LiL22}; this is one reason that developing improved 
approximations for DST has proven challenging. DSF
is hard to approximate to a factor 
$\Omega(2^{\log^{1 - \eps}(n)})$ for any $\eps > 0$ and thus does not 
admit a polylogarithmic approximation \cite{DodisK99}. 
The best known approximation ratios are $O(\nt^{\frac 1 2 +\eps})$ 
\cite{ChekuriEGS11} and $O(n^{2/3 + \eps})$ \cite{BermanBKRY13}. 
An exciting recent line of work has shown that many problems that are hard to 
approximate in general directed graphs are tractable when the input 
graph is \emph{planar} \cite{KawarabayashiS21,FriggstadM23,polymatroid_esa,dsf_soda}.
In planar digraphs, DST admits an $O(\log k)$-approximation \cite{FriggstadM23},
a natural LP relaxation has integrality gap $O(\log^2 k)$ \cite{polymatroid_esa},
and DSF admits an $O(\log^6k)$-approximation \cite{dsf_soda}. 

In this paper, we consider a generalization of DST and DSF to the 
\emph{length-constrained} setting. As the name suggests, these impose an 
additional constraint on the length of paths used to connect terminal pairs. 
Formally, in \emph{Length-Constrained Directed Steiner Tree} (LC-DST) 
and \emph{Length-Constrained Directed Steiner Forest} (LC-DSF), 
we are given an additional \emph{length} function on the edges 
$\ell: E \to \R_+$ and a length constraint $\lc$; the goal is to find a 
min-cost subgraph in which terminal pairs are connected via a path 
of length at most $\lc$. 
There are several practical reasons why connectivity alone may not suffice: 
short paths are desirable for faster transportation and communication, 
improved tolerance to failure, and more. Furthermore, length-constrained 
network design is closely related to several other important problems 
in combinatorial optimization. For instance, \emph{buy-at-bulk} network design 
(closely related to cost-distance and other problems) is a well-studied and 
practical model which places a penalty on the path lengths in the objective 
function instead of imposing strict bounds on path lengths.
Progress in length-constrained network design has led to the resolution of 
several open problems in buy-at-bulk network design \cite{hop_esa,bab_stoc}.
As a result of its practical and theoretical importance, length-constrained 
network design has been very well studied; see \Cref{sec:rel_work} for details. 

In undirected graphs, 
adding length constraints to problems adds significant complexity to 
a problem. As an example, consider the minimum spanning tree problem; 
this is easily solvable in polynomial time. However, the length-constrained 
minimum spanning tree problem is essentially equivalent in approximability
to DST (see \cite{hershkowitz2025planar} for a formal 
reduction). A common way to handle this difficulty is by relaxing the 
length constraint, allowing for some slack in both the 
length and the cost (see \Cref{def:length_bicriteria}). 
In directed graphs, on the other hand, length constraints do not add 
inherent complexity. A simple reduction shows that LC-DST 
is equivalent to DST, as shown in \Cref{fig:hc_dst_red},\footnote{For polynomial edge lengths one can replace each edge $e$ with a 
path of length $\ell(e)$. For larger edge lengths, one can obtain a similar reduction losing a $(1+\eps)$ 
factor in length slack.} 
and one can define a similar reduction from LC-DSF to DSF. 
However, these reductions do not preserve planarity. Given the 
fruitful progress on network design problems in planar digraphs, it is 
natural to ask if LC-DST and LC-DSF admit polylogarithmic approximations 
in planar digraphs. We answer these questions affirmatively.

\begin{figure}
  \begin{center}
    \includegraphics[width=0.7\linewidth]{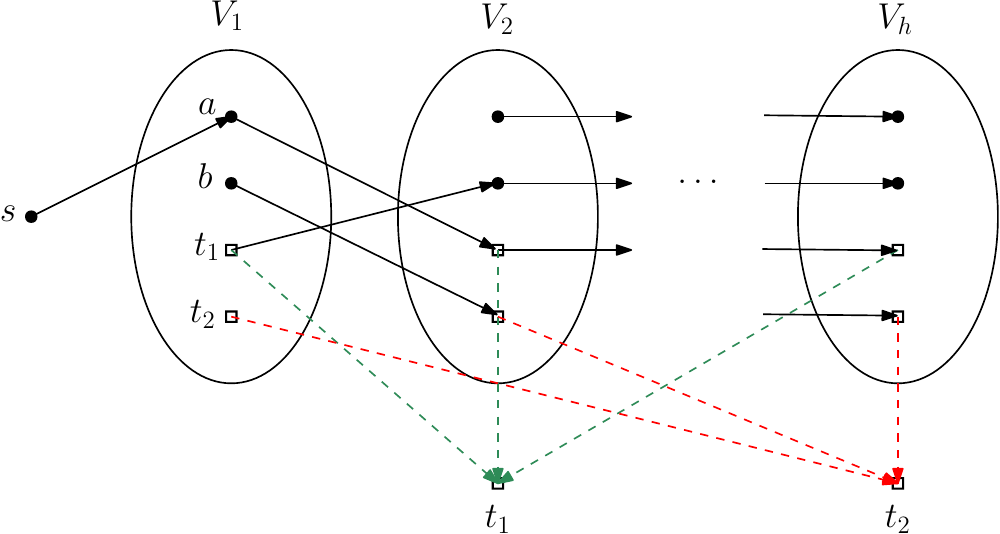}
  \end{center}
  \caption{This example assumes all edge lengths are $1$.
  Suppose the input graph is a path $s,a,t_1,b,t_2$ with 
  terminals $t_1, t_2$. We create $h$ copies of the vertex set 
  $V_1, \dots, V_h$, and for every $(u,v) \in E$ and all $i < h$, we 
  add an edge from the copy of $u$ in $V_i$ to that of 
  $v$ in $V_{i+1}$. The terminals in the resulting DST instance are 
  given at the bottom of the figure.}
  \label{fig:hc_dst_red}
\end{figure}

\subsection{Our Results}

We obtain polylogarithmic \emph{bicriteria} approximation 
algorithms for LC-DST and LC-DSF in planar digraphs. 

\begin{definition}
\label{def:length_bicriteria}
  An $(\alpha, \beta)$-bicriteria approximation for a length constrained 
  network design problem with length constraint $\lc$ is a solution 
  $F \subseteq E$ in which all demand pairs are connected via a path 
  of length at most $\beta \cdot \lc$, and 
  $c(F) \leq \alpha \cdot \opt$, where $\opt$ is the cost of an 
  optimal solution. 
  We refer to $\alpha$ as the \emph{approximation ratio}
  and $\beta$ as the \emph{length slack}.
\end{definition}

The approximation factors for all of our results 
match the best known ratios for DST and DSF in planar digraphs without 
length constraints; all algorithms incur an $O(\log k)$ factor 
in the length slack. We remark that this $O(\log k)$ length slack 
factor is essentially the best known for many problems in length-constrained network design when one seeks a polylogarithmic 
approximation ratio,
including length-constrained MST in undirected graphs. There are known algorithms 
for length-constrained MST with polynomial approximation ratio 
that do not lose any length slack. 
Reducing the length slack to $O(1)$ while ensuring polylogarithmic approximation ratio is an enticing research direction 
with some partial progress discussed further in 
\Cref{sec:rel_work}.
We highlight one such recent work: 
in undirected planar graphs, \cite{hershkowitz2025planar} obtains an 
$O(\log^{1+\eps}(n))$-approximation with $(1+\eps)$ length 
slack. However, it is unknown if this approach would extend to the directed setting, and there appear to be several challenges.

We first consider length-constrained Directed Steiner Tree.  
An approximation algorithm is \emph{LP-competitive} if 
the cost of the given solution is at most $\alpha$ times the 
value of an optimum \emph{fractional} solution to a natural 
LP relaxation. We describe an LP relaxation to LC-DST in 
\Cref{sec:dst}. We obtain two algorithms: the first 
(\Cref{thm:lcdst_integral}) is with respect to the optimum integral solution, 
and the second (\Cref{thm:lcdst_frac}) is LP-competitive. 

\begin{theorem}\label{thm:lcdst_integral}
  There is a bicriteria $(O(\log \nt), O(\log \nt))$-approximation for 
  LC-DST in planar digraphs.
\end{theorem}

\begin{theorem}\label{thm:lcdst_frac}
  There is a bicriteria $(O(\log^2 \nt), O(\log \nt))$ LP-competitive approximation for LC-DST in planar digraphs. 
\end{theorem}

Next, we consider length-constrained Directed Steiner Forest. 
Note that \Cref{thm:lcdsf_main} provides an approximation 
algorithm with respect to the optimum integral solution; obtaining an LP-competitive approximation 
remains open. We note that an LP-competitive approximation 
for planar DSF is not known even without length constraints. 

\begin{theorem}
\label{thm:lcdsf_main}
  There is a bicriteria $(O(\log^6 \nt), O(\log \nt))$-approximation for 
  Length-Constrained Directed Steiner Forest in planar
  digraphs.
\end{theorem}

\subsubsection{Extensions}

We mention two further extensions of our work. Details and formal theorem statements 
are omitted as they are technically involved and require new machinery, see \cite{my_thesis} for details.

\mypara{Buy-at-bulk network design.}
  As mentioned earlier, there are known connections between buy-at-bulk 
  and length-constrained network design problems. Prior work in 
  \emph{undirected} graphs \cite{hop_esa,bab_stoc} demonstrated an LP-based 
  reduction from a buy-at-bulk problem to its corresponding length-constrained 
  problem. The resulting approximation ratio for the buy-at-bulk problem 
  depends on both the approximation and length slack of the algorithm for 
  the length-constrained problem. 
  These arguments extend to directed 
  graphs; thus, one can obtain a polylogarithmic approximation for 
  buy-at-bulk DST in planar digraphs as a 
  corollary of \Cref{thm:lcdst_frac}. 
  In the multicommodity setting, more effort is needed, since 
  \Cref{thm:lcdsf_main} does not yield an LP-competitive solution. 
  Nevertheless, one can obtain a polylogarithmic 
  approximation for buy-at-bulk DSF in planar digraphs 
  by arguing that there exists a low-density \emph{junction tree}. 
  Similar arguments have been used in \cite{chks09,acsz11,hop_esa}.

\mypara{Group, Covering, and Polymatroid Steiner Tree.}
  Chekuri et al. \cite{polymatroid_esa} developed polylogarithmic 
  approximation algorithms for several rooted generalizations of 
  Directed Steiner Tree in planar digraphs. This includes 
  the \emph{Directed Polymatroid Steiner Tree} problem, where 
  a polymatroid is defined on the vertex set, and a feasible tree 
  is required to span a base of the polymatroid. A polymatroid 
  is an integer-valued, normalized, monotone, submodular function. 
  This generalizes several important network design problems, 
  including Group and Covering Steiner Tree. At a high level, 
  the algorithm of \cite{polymatroid_esa} explicitly constructs the recursion tree of the 
  divide-and-conquer algorithm of \cite{FriggstadM23}. One can 
  view this recursion tree as a ``tree embedding'' suitable for 
  solving rooted problems, allowing one to reduce to solving the 
  problem on a tree. One can 
  similarly construct a ``tree embedding'' that respects length 
  constraints, based on the divide-and-conquer algorithm of 
  \Cref{thm:lcdst_integral}. Thus, one can obtain polylogarithmic 
  approximations for length-constrained rooted variants of directed 
  polymatroid, covering, and group problems. 

\subsection{Technical Overview}

The proofs of \Cref{thm:lcdst_integral} and \Cref{thm:lcdst_frac} on 
LC-DST adapt established ideas to the length-constrained setting. 
The main novel contribution of this work is \Cref{thm:lcdsf_main} on 
LC-DSF, as the previous approach for DSF in planar digraphs 
fails to extend. We outline the ideas for all 
three algorithms here.   

\mypara{Approximation for DST in planar digraphs.} We start by giving an overview 
of the divide-and-conquer $O(\log \nt)$-approximation given by 
Friggstad and Mousavi \cite{FriggstadM23}. 
The algorithm uses Thorup's shortest path separator theorem applied to directed graphs:
\begin{lemma}[\cite{FriggstadM23,Thorup04}]\label{lem:dir_sep}
  Let $G$ be a planar directed graph with non-negative edge costs $c(e)$,
  non-negative vertex weights $w(v)$, and a root $s \in V$ such that every vertex 
  in $V$ is reachable from $s$. There exists a polynomial time algorithm 
  to find three shortest dipaths $P_1, P_2, P_3$ starting at $s$ such 
  that every weakly connected component of 
  $G \setminus (P_1 \cup P_2 \cup P_3)$ has at most half the vertex weight of $G$.
\end{lemma}

\begin{figure}[t]
  \centering
  \begin{subfigure}[b]{0.35\textwidth}
    \centering
    \includegraphics[width=\textwidth]{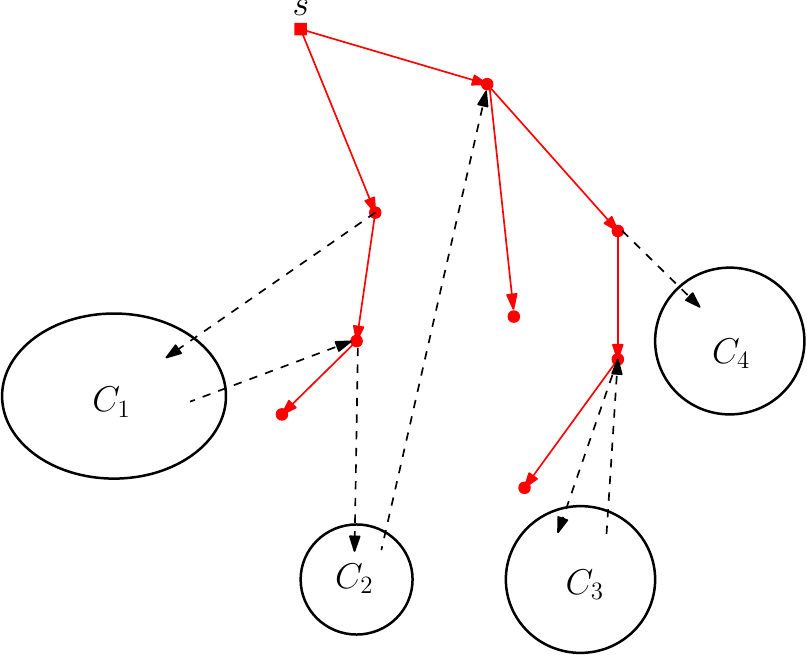}
    \caption{Step 1.}
  \end{subfigure}
  \hfill
  \begin{subfigure}[b]{0.3\textwidth}
    \centering
    \includegraphics[width=\textwidth]{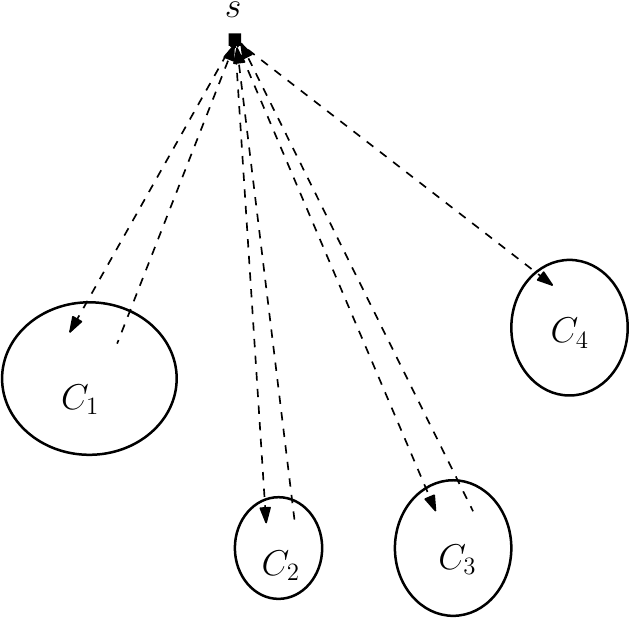}
    \caption{Step 2.}
  \end{subfigure}
  \hfill
  \begin{subfigure}[b]{0.3\textwidth}
    \centering
    \includegraphics[width=\textwidth]{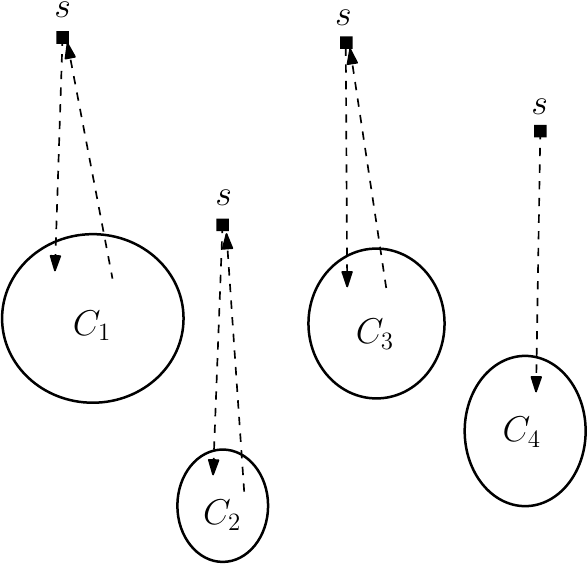}
    \caption{Step 3.}
  \end{subfigure}
  \caption{High-level overview of the divide-and-conquer algorithm. In Step 1, we compute 
  the planar separator (shown in red). Step 2 demonstrates contracting the separator into 
  the root, and Step 3 demonstrates the resulting independent sub-instances.}
  \label{fig:dst_algorithm_steps}
\end{figure}

The high-level idea in \cite{FriggstadM23} is simple; see \Cref{fig:dst_algorithm_steps}. 
Suppose we can guess the optimum 
solution value $\opt$ for a given DST instance. Then, we can remove all vertices 
farther than $\opt$ from $s$ and use \Cref{lem:dir_sep} to find three paths, each 
of cost at most $\opt$, whose removal results in connected components each 
containing at most half the terminals. We shrink the paths into $s$ and recurse 
on the resulting sub-instances. The recursion depth is 
$O(\log \nt)$, which bounds the total cost to $O(\log \nt) \cdot \opt$.
Implementing the guess of $\opt$ in each recursive call is challenging; 
naively guessing requires quasipolynomial time.
\cite{FriggstadM23} obtain a polynomial time algorithm by a refined argument where the 
guess for $\opt$ is folded into the recursion itself. 

\mypara{Handling Length-Constraints and Bounding Integrality Gap.} 
The main challenge in using the approach of \cite{FriggstadM23} 
in \emph{length-constrained} Directed Steiner Tree lies in defining the shortest path tree
on which to apply the planar separator result. Ideally, one would like 
for the tree to capture length-constrained distances, where the $h$-length-constrained
distance from $u$ to $v$ is the min-cost over all
$u$-$v$ paths of length at most $\lc$. Unfortunately, length-constrained 
distances do not form a metric, so one cannot use a standard shortest path 
tree. To circumvent this issue, we use the well-known ``mixture-metric'',
which balances length constraints and costs. This suffices to obtain
\Cref{thm:lcdst_integral}.
We remark that recent work of \cite{hershkowitz2025planar} on length-constrained 
MST in \emph{undirected planar graphs} also uses 
the mixture-metric to apply the separator lemma.
To prove \Cref{thm:lcdst_frac}, we use an idea introduced in 
\cite{polymatroid_esa}, which uses the LP optimum value as the estimate for 
$\opt$ instead of guessing. The algorithm of \cite{polymatroid_esa} extends 
to the length-constrained setting with the mixture-metric in a similar fashion.

\mypara{Directed Steiner Forest via Junction Trees.}
To tackle the multicommodity setting, we build on the \emph{junction tree}
based algorithm of \cite{dsf_soda} for planar DSF. At a high level, 
junction trees allow one to reduce a multicommodity problem to its single-source 
counterpart. This technique has been used in a variety of network 
design problems; it was first explicitly described in the context of non-uniform 
buy-at-bulk network design \cite{hks09,chks09}, though 
it was implicitly used earlier for Directed Steiner Forest in
\cite{Charikaretal99}. The idea is to repeatedly find 
well-structured
\emph{low-density partial solutions}, where a partial solution is feasible 
for a subset of terminal pairs, and the density is the ratio of cost to 
the number of terminal pairs connected. Using a standard
iterative approach for covering problems, one can reduce a multicommodity 
problem to its corresponding min-density partial solution problem, 
losing an $O(\log \nt)$ factor in the approximation ratio. 
We aim to find partial solutions that contain a ``junction'' vertex $v \in V$
through which many pairs connect. This allows us to view the problem 
as a single-source problem with source $v$. 
The approximation for DSF in planar digraphs of \cite{dsf_soda} 
consists of two steps: (1) proving the existence of a low-density 
junction tree, and (2) efficiently finding an approximate min-density 
junction tree. The second step uses a bucketing-and-scaling technique 
introduced by \cite{chks09}, which relies on the existence of an 
LP-competitive single-source algorithm. With \Cref{thm:lcdst_frac}, this extends with little difficulty to the 
length-constrained setting. Our main technical contribution is 
the first step: showing existence of a low-density junction tree,
as length constraints pose several new challenges, and the proof of 
\cite{dsf_soda} does not cleanly extend.
We defer further discussion to \Cref{sec:dsf}.

\subsection{Related Work}
\label{sec:rel_work}

\mypara{Directed Steiner Tree:} Zelikovsky \cite{Zelikovsky97} was the first to address 
the approximability of DST. He obtained an $O(\nt^{\eps})$-approximation for any fixed 
$\eps > 0$ via two ideas. He developed a recursive greedy algorithm and analyzed its 
performance as a function of recursion depth. He then showed that one can 
reduce the problem on a general directed graph to a problem on a depth $d$ DAG 
(via the transitive closure of the original graph) at the loss of an approximation 
factor that depends on $d$. Charikar et al. \cite{Charikaretal99} refined the algorithm 
and analysis in \cite{Zelikovsky97} and combined it with the depth reduction. They 
showed that one can obtain an $O(d^2 \nt^{1/d}\log \nt)$ approximation in $O(n^{O(d)})$-time;
this led to an $O(\log^3 \nt)$-approximation in quasi-polynomial time. 
Subsequently, Grandoni et al. \cite{GrandoniLL22}
improved the approximation to $O(\log^2 \nt/\log \log \nt)$ in quasi-polynomial time 
via a more sophisticated LP-based approach. A different approach that also yields 
the same bound was given by Ghuge and Nagarajan \cite{GhugeN22},
and this is based on a refinement of the recursive greedy algorithm for 
walks in graphs \cite{ChekuriP05}. The approach of \cite{GhugeN22} extends 
to several related problems, including directed single-source buy-at-bulk.

Zosin and Khuller \cite{ZosinK02} showed that the natural cut-based 
LP relaxation for DST has an integrality gap of
$\Omega(\sqrt{\nt})$. However, their example only showed a gap of 
$\Omega(\log n)$ as a function of the number of nodes $n$. There was some hope that 
the integrality gap is polylogarithmic in $n$; however,
\cite{LiL22} recently showed that the gap is $\Omega(n^{\delta})$ for some 
$\delta > 0$ by modifying
the construction in \cite{ZosinK02}. Interestingly, these lower bound examples are 
DAGs with $O(1)$-layers,
for which the recursive-greedy algorithm yields an $O(\log \nt)$-approximation in 
polynomial-time.
Rothvoss \cite{Rothvoss11} showed that $O(\ell)$-levels of the Lasserre SDP 
hierarchy when applied to the standard cut-based LP reduces the integrality 
gap to $O(\ell \log \nt)$ on DAGs with $\ell$ layers.
This was later refined to show that $O(\ell)$-levels of
the Sherali-Adams hierarchy suffices \cite{FriggstadKKLST14}. 
However, both these approaches also require quasi-polynomial time to 
obtain a polylogarithmic approximation. Recent work by 
Laekhanukit \cite{laekhanukit2024integrality}
shows a flow-based LP relaxation that has polylogarithmic integrality 
gap under some strong assumptions on the structure of an 
optimal fractional solution.

\mypara{Directed Steiner Forest.} 
The first nontrivial approximation
for Directed Steiner Forest was an $\tilde O(\nt^{2/3})$-approximation
given by Charikar et al. \cite{Charikaretal99}.
This follows a similar iterative
density-based procedure as the \emph{junction tree} approach; however, they
restrict to trees of a much simpler structure.
This approximation ratio was subsequently improved to $O(\nt^{\frac 1 2 +\eps})$
by Chekuri et al. \cite{ChekuriEGS11}; they showed that given an
instance $(G, D)$ of DSF, there exists a junction tree of density at
most $O(\nt^{1/2})$ times the optimum. They then provided an algorithm
to find a low-density junction tree via height reduction and
Group Steiner Tree rounding. 
DSF has improved approximation ratios when $\nt$ is large. 
Feldman, Kortsarz, and Nutov \cite{feldman_improved_2012} obtained an 
$O(n^{\eps} \cdot \min(n^{4/5},m^{2/3}))$-approximation using a junction-based 
approach. This analysis was refined by Berman et al. \cite{BermanBKRY13} 
using ideas developed for 
finding directed spanners, giving an improved approximation ratio of 
$O(n^{2/3 + \eps})$. For DSF with \emph{uniform} edge costs, 
Chlamt{\'a}{\v{c}} et al. obtained an $O(n^{3/5 + \eps})$-approximation \cite{CDKL20}, 
which was subsequently improved to $O(n^{4/7 + \eps})$ by Abboud and Bodwin 
\cite{abboud2024reachability}.

\mypara{Length-Constrained Network Design.}
Length-constrained network design has been primarily studied in 
undirected graphs. It was first introduced by 
Balakrishnan and Altinkemer \cite{ba92}, and has since been studied 
in both approximation and fast exact algorithms. 
The length-constrained $s$-$t$ shortest path problem is solvable exactly in
polynomial time when edge lengths are polynomially bounded, via the textbook 
Bellman-Ford dynamic programming algorithm. 
However, it is $\NP$-hard for arbitrary lengths \cite{johnson1979computers}.\footnote{The 
length-constrained shortest path problem is often referred to as \emph{Restricted Shortest Path}.} Moreover,
length-constrained MST is $\NP$-hard even in the special case of
hop constraints.\footnote{Hop 
constraints refer to the unit-edge-length setting of length-constrained network design.}
From an approximation algorithms standpoint, until recently, the
majority of the work on length-constrained network design had been
limited to single-source settings, such as $s$-$t$ shortest path
\cite{joksch1966shortest,warburton1987approximation,hassin1992approximation,goel2001efficient,lorenz2001simple,bernstein2012near}, MST 
\cite{dahl98,akgun_tansel11,pirkul_soni03,althaus05,kls05,ravi94,marathe98,HZ25_sosa_mst}, 
Steiner Tree
\cite{kp97,kp06}, and $k$-Steiner Tree \cite{hks09,ks11}. 
While $s$-$t$ shortest path admits an FPTAS \cite{hassin1992approximation}, other 
length-constrained problems have strong hardness results under 
\emph{strict} length constraints. Length-constrained 
MST is as hard as Directed Steiner Tree, 
and the length-constrained variant of Steiner forest has no
$o(2^{\log^{1-\eps}(n)})$-approximation \cite{dkr16}. These problems are 
more tractable when length constraints are relaxed. 
A simple merging algorithm \cite{marathe98}
yields an $(O(\log n), O(\log n))$-bicriteria approximation
for length-constrained MST and Steiner Tree. Recent progress 
by Hershkowitz and Huang allows one to reduce the length slack to 
$O(\log n/\log \log n)$ while increasing the approximation ratio by 
a polylogarithmic factor \cite{HZ25_sosa_mst}. The multicommodity 
setting has been more challenging.
For length-constrained Steiner Forest, polylogarithmic bicriteria approximations 
were implicitly known via buy-at-bulk \cite{chks09,kortsarz_nutov11}. 
There has also been a line of work on hop-constrained \emph{tree embeddings}, 
which have led to bicriteria approximations for hop-constrained 
Steiner Forest and various related problems \cite{hop_distance21,filtser22}.
Chekuri and Jain \cite{hop_esa}, building on hop-constrained oblivious 
routing tools of Ghaffari, Haeupler, and Zuzic \cite{hop_congestion21}, obtained 
LP-competitive algorithms for multicommodity problems in buy-at-bulk and 
length-constrained settings. In directed graphs, a recent work 
obtains polynomial factor approximations for Directed Buy-at-Bulk and 
related spanner problems \cite{grigorescu2025directed}; 
these approximation ratios essentially match the best known factors for 
DSF in general digraphs.

\mypara{Planar and Minor-Free Graphs.}
Improved approximation ratios have been obtained for several problems
in special classes of graphs, such as planar and minor-free graphs.
We first discuss undirected graphs.
In planar graphs, Steiner Tree admits a PTAS \cite{BorradaileKM09};
this was later extended
to a PTAS for Steiner Forest in graphs of bounded genus
\cite{BateniHM11}.
Recently, \cite{cohen-addad_bypassing_2022} obtained a QPTAS for Steiner Tree in
minor-free graphs. Furthermore, although the node-weighted variant
of Steiner Tree captures Set Cover in general graphs,
there exists a constant factor approximation in
planar graphs, and more generally, in any proper minor-closed
graph family \cite{DemaineHK14}. In directed graphs, several 
recent results have built on tools given by Thorup \cite{Thorup04}
to obtain approximation algorithms in planar digraphs. 
Kawarabayashi and Sidiropoulos \cite{KawarabayashiS21} gave a polylogarithmic 
upper bound on the multicommodity flow-cut gap in planar digraphs;
this was derandomized and extended to the node-weighted setting by 
\cite{chekuri2026node}. 
As discussed earlier, Friggstad and Mousavi \cite{FriggstadM23} 
obtained an $O(\log k)$-approximation for DST in planar digraphs. 
This was extended by Chekuri et al. \cite{polymatroid_esa} to an 
LP-competitive $O(\log^2 k)$-approximation. Chekuri et al.
\cite{polymatroid_esa} also obtained polylogarithmic approximations for 
several rooted variants of DST in planar digraphs, such as group, covering, 
polymatroid, and multirooted versions. Chekuri and Jain \cite{dsf_soda} 
then obtained an $O(\log^6 k)$-approximation for DSF in planar digraphs.
In minor-free input graphs that are \emph{quasi-bipartite}, 
Friggstad and Mousavi \cite{FriggstadM21}
obtained a constant-factor approximation for DST.\footnote{The 
input is quasi-bipartite if 
there are no edges between any two non-terminal nodes.}
\section{Preliminaries}
\label{sec:prelim}

\subsection{Notation}

Let $G = (V, E)$ be a directed graph with edge costs $c:E\rightarrow\R_+$. 
For $E'\subseteq E$, we denote $c(E') = \sum_{e \in E'}c(e)$. 
We assume all edge costs $c(e) \geq 1$ and are polynomially bounded in $n$. 
This can be done with a
$(1 + o(1))$ factor loss in approximation, 
by guessing the cost of the optimal solution $\opt$, contracting edges with
cost much smaller than $\opt$, and scaling appropriately. 

For $S \subseteq V$, we use $E[S]$ to denote the set of edges of $E$ with both 
endpoints in $S$, $G[S]$ to denote the subgraph $(S, E[S])$ \emph{induced by $S$} 
in $G$, and $\delta^+(S) = \{(u, v) \in E : u \in S, v \not \in S\}$ to denote 
the \emph{out-cut} of $S$. 
The \emph{height} of an out-tree $T$ rooted at $s$ is the maximum 
over all $v \in V(T)$ of the number of edges on the unique $s$-$v$ path. 
The \emph{size} of the tree is the number of vertices in the tree.
To \emph{contract} an edge $(u, v)$ is to merge its endpoints
into a single vertex $w_{uv}$, replacing all edges incident to 
$u$ or $v$ with edges incident to $w_{uv}$ and removing resulting 
self-loops.
For a subgraph $G' \subseteq G$, we use $G/G'$ to denote the graph obtained by 
contracting every edge of $G'$ and $G - G'$ to denote the graph obtained by deleting 
every edge in $G'$. A graph $H$ obtained from $G$ via edge contractions 
and edge or vertex deletions is called a \emph{minor} of $G$. 
It is
easy to see that if $G$ is planar, then any minor of $G$ is planar as well.
A weakly connected component of $G$ is a maximal connected component 
of the underlying 
undirected graph obtained from $G$ by ignoring the edge orientations. 
We call this undirected graph the \emph{undirected version} of $G$. 

The length $\ell(p)$ of a path $p$ 
is the sum of the lengths of its edges. More generally, we write 
$\ell(F) := \sum_{e \in F} \ell(e)$ for any $F \subseteq E$.
For any subset of edges $F \subseteq E$ and any pair of vertices 
$u,v \in V$, we let $\ell_F(u,v)$ denote the length of the 
shortest-length $u$-$v$ path in $F$. 
We generalize this notation to 
sets: for $A, B \subseteq V$, $\ell_F(A, B)$ is the length of the shortest-length 
path in $F$ from $A$ to $B$. We drop $F$ when it is clear from context.
For a tree $T$ and vertices $u,v \in V(T)$, we write $P_T(u,v)$ 
to denote the unique $u$-$v$ path in $T$. 

\subsection{Mixture Metric}

We write $d^{(\lc)}(s,t)$ to denote the cost of the shortest $\lc$-length constrained 
$s$-$t$ path. Computing $d^{(\lc)}(s,t)$ is $\NP$-hard \cite{johnson1979computers}; however, it can be 
approximated to a $(1+\eps)$-factor \cite{hassin1992approximation}. Unlike traditional graph distances,
$d^{(\lc)}$ does not form a metric. We instead consider 
the \emph{mixture metric}; this is a standard notion to 
facilitate working with length constraints, 
and was highlighted in recent work in length-constrained network design by 
\cite{hop_distance21}.

Given a length constraint $\lc$ and a cost bound $\gamma$, we define the mixture 
metric as follows: for each $e \in E$, the \emph{mixture cost} 
function is
$c'(e) = c(e) + \frac{\gamma}{\lc} \ell(e)$. For all $s,t \in V$, let 
$d'(s,t)$ be the cost of the shortest $s$-$t$ path with respect to $c'$. 
It is not difficult to see that $d'$ forms a metric. 
\Cref{claim:mixture_metric} demonstrates how $d'$ approximates $d^{(\lc)}$.

\begin{restatable}{claim}{mixturemetric}
\label{claim:mixture_metric}
  Let $d'$ be the mixture metric with length constraint $\lc$ and cost bound $\gamma$. 
  \begin{enumerate}
    \item $d'(s,t) \leq d^{(\lc)}(s,t) + \gamma$;
    \item $d^{(\lc')}(s,t) \leq d'(s,t)$, where $\lc' = \frac{d'(s,t)}{\gamma} \cdot \lc$.
  \end{enumerate}
\end{restatable}
\begin{proof}
  Fix $s,t \in V$. To prove (1), let $P$ be the shortest $\lc$-length bounded 
  $s$-$t$ path in $G$. Then,
  $d'(s,t) \leq \sum_{e \in P} c'(e) = c(P) + \frac{\gamma}{\lc} \ell(P) 
  \leq d^{(\lc)}(s,t) + \gamma$, since $\ell(P) \leq \lc$.
  To prove (2), let $P'$ be the shortest $s$-$t$ path with respect to $c'$,
  so $d'(s,t) = c'(P')$. Then $\ell(P') \leq d'(s,t) \cdot 
  \frac{\lc}{\gamma} = \lc'$. Thus, $d^{(\lc')}(s,t) \leq c(P') \leq d'(s,t)$. 
\end{proof}

\subsection{Planar Separator}

The key property of planarity used in our algorithms is the fact that planar 
graphs have good shortest path separators, as shown in 
\Cref{lem:dir_sep}. We state a more general lemma 
on \emph{undirected graphs}. 
This lemma, 
essentially proved by Lipton and Tarjan \cite{LiptonTarjan79}, but given explicitly 
by Thorup \cite{Thorup04}, states that every spanning tree on any planar
graph contains three paths whose removal results in connected
components each containing at most half the number of vertices. 
The formal statement is as follows.

\begin{restatable}[\cite{LiptonTarjan79,Thorup04}]{lemma}{undirplanarsep}
  \label{lem:undir_sep}
  Let $G$ be an undirected connected graph with non-negative edge costs $c(e)$,
  non-negative vertex weights $w(u)$, and a spanning tree $T$ rooted at $v \in V$. 
  There exists a 
  polynomial time algorithm to find three vertices $u_1, u_2, u_3 \in V$ 
  such that each component of $G \setminus (P_T(v, u_1) \cup P_T(v, u_2) \cup 
  P_T(v, u_3))$ has at most half the vertex weight of $G$.
\end{restatable}

Note that one can recover the directed variant 
(\Cref{lem:dir_sep}) by considering 
the spanning tree $T$ to be a shortest path tree rooted at $v$.
\section{Length-Constrained Directed Steiner Tree}
\label{sec:dst}

In this section, we consider length-constrained Directed Steiner Tree. 
We provide an $O(\log \nt)$-approximation in \Cref{sec:dst_integral}, and an 
LP-based $O(\log^2 \nt)$-approximation in \Cref{sec:dst_frac}. 

First, we define a subroutine $\algsep$ (\Cref{algo:prune_sep})
which will be used in both the integral and LP-based algorithms. This subroutine 
takes as input a graph $G = (V, E)$ with edge costs $c(e)$, a
root $s \in V$, terminals $S \subseteq V$, and a guess $\gamma$ for the cost of the 
optimal solution.
$\algsep(G, s, S, \gamma)$ removes all vertices further than $\gamma$ away 
from $s$, and uses \Cref{lem:dir_sep} on the resulting graph with vertex weights 
set to 1 on the terminals and 0 elsewhere. This yields a \emph{planar separator} 
$P := P_1 \cup P_2 \cup P_3$ in which each resulting component of $G \setminus P$ has at 
most half the terminals.  
The subroutine contracts $P$ into $s$; each component of $G \setminus P$ 
corresponds to a new subinstance induced by the terminals in that component 
along with $s$. $\algsep(G, s, S, \gamma)$ returns $P$ 
along with the subinstances for each component. 

\begin{algorithm}[H]
\caption{$\algsep(G, s, S, \gamma)$}
\label{algo:prune_sep}
\begin{algorithmic}
    \State Delete all vertices $v \in V$ with $d_G(s, v) > \gamma$
    \State Let $P := P_1 \cup P_2 \cup P_3$ be given by
    \Cref{lem:dir_sep} with weights $w(v) = \mathbbm{1}_{v \in S}$ for $v \in V(G)$
    \State Let $G_P$ be obtained from $G$ by contracting $P$ into $s$
    \State Let $C_1, \dots, C_{\tau}$ be components of $G \setminus P$, and let
    $G_i \gets G_P[C_i \cup \{s\}]$
    \State \Return $(P, (G_1, C_1), \dots, (G_{\tau}, C_{\tau}))$
\end{algorithmic}
\end{algorithm}

\subsection{Approximation with respect to Integral Optimum}
\label{sec:dst_integral}

We follow the divide-and-conquer algorithm given by
Friggstad and Mousavi \cite{FriggstadM23}, adapted to handle length constraints. 
Recall the high level approach:
the algorithm branches on two cases and returns the cheaper
solution. In the first branch, it recurses with a halved guess for 
$\opt$; this handles the case where $\opt \leq \gamma/2$. 
In the second branch, it
computes a planar separator $P$ under the \emph{mixture cost} $c'$, 
and recursively solves each resulting component
independently. Using the mixture cost for the separator is the main 
adaptation for length constraints, which 
ensures that the separator paths have 
bounded cost \emph{and} bounded length.
For clarity, we write the algorithm as though $d^{(\lc)}(s,t)$ can be
computed exactly; replacing this with a $(1+\eps)$-approximation
does not affect the approximation guarantees.
The algorithm is given by \Cref{alg:hDST}.

\begin{algorithm}[H]
\caption{$\hDST(I, \gamma)$: an instance $I = (G, s, S, \lc)$ and
a ``guess'' $\gamma$ for $\opt$}
\label{alg:hDST}
\begin{algorithmic}
    \If{$\exists t \in S$ such that $d^{(\lc)}(s, t) > \gamma$}
        \State \Return Infeasible ($\infty$)
    \EndIf
    \If{$|S| = 1$}
        \State \Return min-cost $\lc$-length-constrained $s$-$t$ path in $G$,
        where $t$ is the terminal in $S$
    \EndIf
    \Statex \textbf{\textsc{// Recursive Step 1:}}
    \State $F^1 \gets \hDST(I, \gamma/2)$
    \Statex \textbf{\textsc{// Recursive Step 2:}}
    \State Define mixture cost $c'(e) = c(e) + \frac{\gamma}{\lc} \ell(e)$
    \State $(P, (G_1, C_1), \dots, (G_\tau, C_\tau)) \gets \algsep(G = (V, E, c'), s, S, 2\gamma)$
    \State $F^2 \gets P$
    \For{$i \in [\tau]$}
        \State $F_i \gets \hDST(I_i := (G_i, s, S \cap C_i, \lc), \gamma)$
        \State $F^2 \gets F^2 \cup F_i$
    \EndFor
    \State \Return $\argmin_{F \in \{F^1, F^2\}} c(F)$
\end{algorithmic}
\end{algorithm}

First, we argue that the algorithm runs in polynomial time. This follows 
directly from \cite{FriggstadM23}; the proof is given for completeness.
\begin{lemma}
\label{lem:dst_int_runtime}
  \Cref{alg:hDST} runs in polynomial time.
\end{lemma}
\begin{proof}
  We prove by induction on $\gamma + |S|$ that the number of recursive 
  calls (including the current step) is at most $f(\gamma, |S|) := 
  \gamma \cdot |S|^3$. The base case when $|S| = 1$ is clear. 
  For the inductive step, let $C_1, \dots, C_\tau$ be the components 
  constructed in Recursive Step 2, and let $\nt_i = |C_i \cap S|$ 
  for each $i \in [\tau]$.  
  By induction, the number of recursive calls is 
  \begin{align*}
    1+ f(\gamma/2, |S|) + \sum_{i \in [\tau]} f(\gamma, \nt_i) 
    \leq 1 + \frac 1 2 \gamma |S|^3 + \sum_{i \in [\tau]} \gamma \nt_i^3.
  \end{align*}
  Since the number of terminals is at least halved at 
  each step, $\nt_i \leq |S|/2$. Thus, $\sum_{i \in [\tau]} \gamma \cdot \nt_i^3
  \leq \gamma (|S|/2)^2 \sum_{i \in [\tau]} \nt_i$. Since the components 
  $C_1, \dots, C_\tau$ are disjoint, $\sum_{i \in [\tau]} \nt_i \leq |S|$. 
  Therefore, the number of recursive calls is at most 
  $1 + \frac 1 2 \gamma |S|^3 + \frac 1 4 \gamma |S|^3 \leq 
  \gamma|S|^3$, since $|S| \geq 2$. 
\end{proof}

Next, we prove feasibility and bound the cost of the given solution. 
We start with a simple property about the planar separator. 

\begin{claim}
\label{claim:dst_int_separator}
  The planar separator $P$ computed by $\hDST(I, \gamma)$ for any 
  feasible instance satisfies the following: (i) $c(P) \leq 6\gamma$,
  (ii) $\ell(P) \leq 6\lc$. 
\end{claim}
\begin{proof}
  Let $P = P_1 \cup P_2 \cup P_3$ be the planar separator computed by 
  $\hDST(I, \gamma)$. Recall that each $P_i$, for $i \in [3]$,
  is a shortest path from 
  $s$ with respect to $c'$. Since all vertices $u \in V$ with 
  $d_{G, c'}(s, u) > 2\gamma$ are removed, each path $P_i$ must have 
  $c'(P_i) = c(P_i) + \frac \gamma \lc \ell(P_i) \leq 2\gamma$. 
  Since both terms are non-negative, $c(P_i) \leq 2\gamma$, so 
  $c(P) \leq 6\gamma$. Similarly, 
  $\frac \gamma \lc \ell(P_i) \leq 2\gamma$, so $\ell(P_i) \leq 2\lc$ 
  and thus $\ell(P) \leq 6\lc$. 
\end{proof}

\begin{lemma}\label{lem:dst_int_feasible}
  Let $F$ be the solution returned by $\hDST(I, \gamma)$ for a feasible 
  instance $(I, \gamma)$. Then, $F$ contains 
  an $s$-$t$ path of length at most $O(\log |S|) \cdot \lc$ for each 
  terminal $t \in S$.
\end{lemma}
\begin{proof}
  We prove by induction on $\gamma + |S|$ that if the instance $(I, \gamma)$ 
  is feasible, then $F := \hDST(I, \gamma)$ contains a path from $s$ to 
  $t$ of length at most $6\lc (\log |S|+1)$, for each $t \in S$. 
  The base case is clear; if there is one terminal, then $F$ is an 
  $s$-$t$ path of length $\lc$. 

  We consider the two recursive cases separately. For the first recursive 
  case, if $(I, \gamma/2)$ is infeasible, then $c(F^1)$ is 
  infinite and $F \neq F^1$. Else, by induction, $F^1$ contains an 
  $s$-$t$ path of length at most $6\lc (\log |S|+1)$ for each $t \in S$. 

  Next, we consider the second recursive case. Fix $t \in S$.
  Observe that since 
  $(I, \gamma)$ is feasible, $d^{(\lc)}(s, t) \leq \gamma$, so there exists 
  an $s$-$t$ path $q$ of length at most $\lc$ and cost at most $\gamma$. 
  In particular, $c'(q) = c(q) + \frac \gamma \lc \ell(q) \leq 2\gamma$, 
  so $t$ is not deleted in the pruning step. 
  Therefore, $t$ is either in the planar separator $P$ or 
  in some component $C_i$ of $G \setminus P$. If $t \in P$, then 
  $P$ contains an $s$-$t$ path of length at most $\ell(P) \leq 6\lc$,
  by \Cref{claim:dst_int_separator}. If $t \in C_i$, then since the recursive 
  sub-instance contracts $P$ into $s$, $C_i$ contains a path from 
  $\{s\} \cup P$ to $t$ of length at most $6\lc (\log|S \cap C_i| + 1)$, 
  by induction. Therefore, $F^2$ contains an $s$-$t$ path of length 
  at most 
  \begin{align*}
    \ell(P) + 6\lc (\log|S \cap C_i| + 1) 
    \leq 6\lc + 6\lc(\log(|S|/2) + 1) = 6\lc + 6\lc \log |S| = 6\lc(\log |S| + 1).
  \end{align*}
  This concludes the inductive argument and the proof of the lemma. 
\end{proof}

\begin{lemma}\label{lem:dst_int_cost}
  The cost of the solution returned by $\hDST(I, \gamma)$ for a feasible 
  instance $(I, \gamma)$ is at most $O(\log |S|) \cdot \opt$. 
\end{lemma}
\begin{proof}
  Let $F$ be the solution given by $\hDST(I, \gamma)$. 
  We prove by induction on $\gamma + |S|$ that 
  $c(F) \leq 12\gamma(\log|S| + 1)$. 
  We can assume that $\gamma \leq 2\opt$; if not, then $\gamma/2 \geq \opt$,
  so we can apply induction on $c(F^1)$ and use the fact that 
  $c(F) \leq c(F^1)$. 

  We bound the cost of $F^2$; this suffices to complete the proof,
  since $c(F) \leq c(F^2)$. Let $P$ be the separator computed at this step 
  and let $C_1, \dots, C_\tau$ be the resulting components of $G \setminus P$.
  For $i \in [\tau]$, we let $\opt_i$ denote the cost of the optimal solution 
  on $I_i$. Note that $\sum_{i \in [\tau]} \opt_i \leq \opt$, since 
  the optimal solution on $G$ restricted to each $C_i$ is feasible for 
  $I_i$, and all $C_i$ are disjoint from each other.
  By \Cref{claim:dst_int_separator},
  $c(P) \leq 6\gamma \leq 12\opt$. By induction,
  \begin{align*}
    c(F^2) &= c(P) + \sum_{i \in [\tau]} c(F_i) 
    \leq 12\opt + \sum_{i \in [\tau]} 12(\log|S \cap C_i| + 1)\opt_i \\
    &\leq 12\opt + \sum_{i \in [\tau]} 12(\log |S|) \opt_i 
    \leq 12(\log|S| + 1) \opt. \qedhere
  \end{align*}
\end{proof}

\Cref{thm:lcdst_integral} follows from 
\Cref{lem:dst_int_runtime}, \Cref{lem:dst_int_feasible},
and \Cref{lem:dst_int_cost} by applying \Cref{alg:hDST} to 
the given input instance $I$, with $\gamma$ chosen to be any 
upper bound on $\opt$; e.g. $c(E)$. 

\subsection{LP-Based Approximation}
\label{sec:dst_frac}

The goal of this section is to obtain a polylogarithmic upper bound 
on the integrality gap of a natural LP relaxation for LC-DST in 
planar digraphs, defined as follows. For each terminal $t \in S$, 
let $\calP_t^{\lc}$ denote the set of all $s$-$t$ paths in $G$ of length at most $\lc$.
The variables $x_e$ indicate whether an edge $e$ is in the final solution, and 
variables $f_p$ indicate how much $s$-$t$ flow is sent on path $p$. 

\begin{equation}
\label{LC-DST-LP}
\tag{LC-DST-LP}
\begin{aligned}
  \min\quad \sum_{e\in E}c(e)x_e& \\
  s.t.\quad \sum_{p \in \calP_t^{\lc}} f_p &\geq 1 \qquad~
  \forall t \in S \\
  \sum_{p \in \calP_t^{\lc}, e \in p} f_p &\leq x_e \qquad 
  \forall e \in E, t \in S \\
  x_e, f_p &\geq 0 \qquad~
  \forall e\in E, p \in \cup_{t \in S} \calP_t^\lc
\end{aligned}
\end{equation}

The LP contains an exponential number of variables but can be solved approximately 
(up to a $(1+\eps)$ factor) via separation oracle on the dual.

We provide an LP-based algorithm following a similar framework as above;  
here, we use the LP solution as the ``guess'' for $\opt$. This is the 
same approach used by \cite{polymatroid_esa} to obtain an LP-based approximation 
for DST in planar graphs.
The algorithm $\hfracDST(I, (x, f))$ is as follows, where 
$I = (G, s, S, \lc)$ is a problem instance and $(x, f)$ is a feasible 
LP solution to the relaxation for \ref{LC-DST-LP}.

\begin{figure}[t]
\centering
\includegraphics[width=0.5\linewidth]{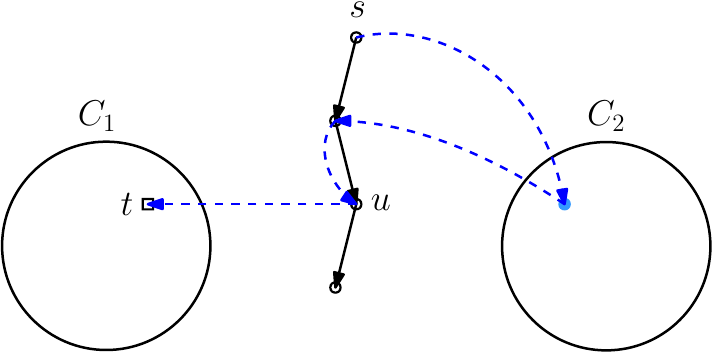}
\caption{Consider the $s$-$t$ path $p$ shown in blue dashed lines;
here the separator $P$ is black and bolded. The portion of the path 
from $u$ to $t$ corresponds to $p'$.}
\label{fig:lcdst_paths}
\end{figure}

\begin{itemize}
  \item \textbf{Base Case.} If $S = \{t\}$, we return the min-cost $\lc$-length 
  constrained $s$-$t$ path.
  \item \textbf{Prune and Separate.} We set 
  $\gamma = 2\log|S| \sum_{e \in E} c(e) x_e$ to be the ``guess'' for 
  $\opt$. As in \Cref{alg:hDST}, we define the mixture edge cost function 
  as $c'(e) = c(e) + \frac \gamma \lc \ell(e)$.
  Let $(P, (G_1, C_1), \dots, (G_\tau, C_\tau))$ denote the output of 
  $\algsep(G = (V, E, c'), s, S, 2\gamma)$. 
  \item \textbf{Recursive Step.} For each component $i \in [\tau]$, 
  we define the recursive sub-instance 
  $I_i = (G_i, s, S \cap C_i, \lc)$. In order to construct 
  a feasible LP solution for $I_i$, we first scale up the solution;
  let $(\bar x, \bar f) = \left(1 + \frac 1 {\log |S|}\right) (x, f)$. 
  We set $x_e^i = \bar x_e$ for all $e \in G_i$. 

  For the path variables, fix a terminal $t \in S \cap C_i$. For each 
  path $p \in \calP^{\lc}_t$ (in the graph $G$ \emph{before} contracting $P$ 
  into $s$), let $p'$ denote the corresponding contracted path with 
  resulting cycles deleted. Note that $\ell(p') \leq \lc$ since contraction 
  cannot increase path length, so $p'$ is a valid $s$-$t$ path in the 
  new sub-instance. We define $\calQ(p')$ as the set of all paths in 
  $\calP^{\lc}_t$ that map to $p'$ after contraction and cycle removal;
  see \Cref{fig:lcdst_paths}.
  We set $f^i_{p'} = \sum_{p \in \calQ(p')} \bar f_p$ for all such $p'$ fully 
  contained in $G_i$. We recursively compute 
  $F_i = \hfracDST(I_i, (x^i, f^i))$.
  \item Return $P \cup \left(\cup_i F_i\right)$
\end{itemize}

It is clear that this algorithm runs in polynomial time.
We show that each recursive call is valid by showing that the 
constructed LP solutions are feasible for the subinstances. 

\begin{lemma}
  For each $i \in [\tau]$, 
  $(x^i, f^i)$ is a feasible LP solution on instance $I_i$. 
\end{lemma}
\begin{proof}
  Fix $i \in [\tau]$. We show that $(x^i, f^i)$ satisfies all LP constraints. 
  Fix $t \in S \cap C_i$. To 
  distinguish between the set of $\lc$-length-constrained $s$-$t$ paths 
  in the new instance $I_i$ versus in the original input instance $I$, we 
  write $\calP_t^i$ versus $\calP_t$; the length-constraint 
  $\lc$ is dropped from this notation for visual clarity. 

  We start by showing that 
  $\sum_{p \in \calP_t, c'(p) \leq 2\gamma} \bar f_p \geq 1$; 
  that is, scaling up by a $1 + 1/\log|S|$ factor compensates for vertices 
  deleted during the pruning step.  
  Consider a path $p \in \calP_t$ with $c'(p) > 2\gamma$. 
  Since $p \in \calP_t$, $\ell(p) \leq \lc$. Thus, 
  $c(p) = c'(p) - \frac \gamma \lc \ell(p) > 2\gamma - \gamma = \gamma$. 
  Observe that since $(x, f)$ satisfies the LP constraints, 
  \begin{align*}
    \sum_{e} c(e)x_e 
    \geq \sum_e c(e) \sum_{p \in \calP_t: e \in p} f_p 
    = \sum_{p \in \calP_t} f_p c(p) 
    > \gamma \sum_{p \in \calP_t: c(p) > \gamma} f_p.
  \end{align*}
  In particular, $\sum_{p \in \calP_t: c(p) > \gamma} f_p
  < \frac 1 \gamma \sum_{e} c(e) x_e = 1/(2 \log |S|)$. Furthermore, 
  since $(x, f)$ is feasible for $I$, $\sum_{p \in \calP_t} f_p \geq 1$. 
  Therefore, 
  \begin{align*}
    \sum_{p \in \calP_t, c'(p) \leq 2\gamma} \bar f_p
    \geq (1 + \frac 1 {\log|S|}) \sum_{p \in \calP_t: c(p) \leq \gamma} f_p 
    \geq \left(1 + \frac 1 {\log|S|}\right)\left(1 - \frac 1{2\log|S|}\right) 
    \geq 1.  
  \end{align*} 
  Next, we show that $\sum_{p \in \calP^i_t} f_{p}^i = 
  \sum_{p \in \calP_t, c'(p) \leq 2\gamma} \bar f_p$. Let $p \in \calP_t$, 
  and let $p'$ denote the corresponding contracted path. 
  Since $c'(p) \leq 2\gamma$, no vertices of $p$ are deleted during 
  the pruning step. We want to 
  show that $p'$ is contained in $G_i$. Let $v$ be the 
  \emph{last} vertex on $p$ that is not in $C_i$; note that $p$ must 
  end in $C_i$ since $p$ is an $s$-$t$ path and $t \in C_i$. 
  Since the components $C_1, \dots, C_\tau$ are weakly connected in 
  $G \setminus P$, there must be no edges between them. 
  Thus, $v \in P \cup \{s\}$. Therefore, when $P$ is contracted 
  into $s$, the path $p[s,v]$ corresponds to a walk starting and ending 
  at the contracted supernode $s$. This implies that $p' = p[v,t]$, 
  which in the contracted graph, is contained in 
  $G_i$. 
  Therefore, $p \in \calQ(p')$ for some $p' \in \calP_t^i$, so 
  $\sum_{p \in \calP_t, c'(p) \leq 2\gamma} \bar f_p 
  \leq \sum_{p' \in \calP_t^i} \sum_{p \in \calQ(p')} \bar f_p
  = \sum_{p' \in \calP_t^i} f_{p'}^i$, as desired. Thus, 
  $\sum_{p \in \calP^i_t} f_{p}^i \geq 1$, so the first set of 
  LP constraints is satisfied. 

  Proving that the second set of LP constraints are satisfied is more 
  straightforward. We want to show that for all edges $e \in E(G[C_i \cup \{s\}])$,
  $\sum_{p \in \calP_t^i, e \in p} f_p^i \leq x_e^i$. 
  \begin{align*}
    \sum_{p' \in \calP^i_t, e \in p'} f_{p'}^i 
    = \sum_{p' \in \calP^i_t, e \in p'} \sum_{p \in \calQ(p')} \bar f_p
    \leq \sum_{p \in \calP_t, e \in p} \bar f_p,
  \end{align*}
  since the sets $\calQ(p')$ are disjoint. To conclude, since $(x, f)$ 
  is a feasible LP solution,
  \begin{align*}
    \sum_{p \in \calP_t, e \in p} \bar f_p
    = (1 + \frac 1 {\log|S|}) \sum_{p \in \calP_t, e \in p} f_p
    \leq (1 + \frac 1 {\log|S|}) x_e = x_e^i. &\qedhere
  \end{align*}
\end{proof}

The ``prune and separate'' portion of \hfracDST\ is the same as 
in \hDST\ (\Cref{alg:hDST}). Therefore, 
\Cref{claim:dst_int_separator} holds for \hfracDST. 
It is not difficult to see that the 
resulting solution $F = \hfracDST(I, (x, f))$ contains 
an $s$-$t$ path of length at most $O(\log |S|) \cdot \lc$ 
for each terminal $t \in S$, using reasoning analogous to 
that of \Cref{lem:dst_int_feasible}. 
We bound the approximation ratio in 
the following lemma, which concludes the proof of \Cref{thm:lcdst_frac}.

\begin{lemma}
\label{lem:dst_frac_cost}
  If $(x, f)$ is a feasible LP solution to the 
  LC-DST instance $I$, then
  the cost of the solution returned by $\hfracDST(I, (x, f))$ is at most 
  $O(\log^2 \nt) \sum_{e} c(e) x_e$.
\end{lemma}
\begin{proof}
  Let $F$ be the solution returned by $\hfracDST(I, (x, f))$. 
  We let $X = \sum_{e} c(e)x_e$.
  We prove by induction on the number of terminals that 
  $c(F) \leq 12(\log|S| + 1)^2 X$. 

  In the base case, it is easy to see that the cost of the min-cost 
  $\lc$-length-bounded $s$-$t$ path is at most $X$. For the recursive step, 
  for each $i \in [\tau]$, let $X_i$ denote the value of the LP solution 
  $(x^i, f^i)$, and let $F_i$ denote the output of 
  $\hfracDST(I_i, (x^i, f^i))$. 
  By \Cref{claim:dst_int_separator}, 
  $c(P) \leq 6\gamma = 12X \log |S|$. By induction, 
  \begin{align*}
    c(F) = c(P) + \sum_{i \in [\tau]} c(F_i) 
    \leq 12X \log |S| + \sum_{i \in [\tau]} 12X_i (\log|S \cap C_i| + 1)^2.
  \end{align*}
  Since the number of terminals is halved at each recursive step, 
  $\log|S \cap C_i| \leq \log(|S|/2) \leq \log|S| - 1$. 
  Furthermore, since $C_1, \dots, C_\tau$ are disjoint, 
  $\sum_{i \in [\tau]} X_i \leq (1 + \frac 1 {\log|S|}) X$. Thus
  \begin{align*}
    c(F) &\leq 12\log |S| X + \sum_{i \in [\tau]} 12(\log |S|)^2 X_i 
    \leq 12\log |S| X + 12(\log |S|)^2 (1 + \frac 1 {\log|S|})  X \\
    &= 12X (\log^2|S| + 2\log|S|) 
    \leq 12X (\log|S|+1)^2. 
  \end{align*}
  This concludes the inductive argument; thus, 
  $c(F) \leq O(\log^2\nt) \sum_{e \in E} c(e)x_e$.
\end{proof}
\section{Length-Constrained Directed Steiner Forest}
\label{sec:dsf}

In this section, we consider length-constrained Directed Steiner Forest 
in planar digraphs (called length-constrained planar DSF). We are given as input 
a directed graph $G = (V, E)$ with non-negative edge costs $c(e)$, lengths 
$\ell(e)$, length constraint $h$, and demand pairs $D = \{(s_i, t_i)\}_{i \in [\nt]}$.
The goal is to find a min-cost subgraph 
$H \subseteq G$ such that for each $(s_i, t_i) \in D$, $H$ contains an $s_i$-$t_i$ path of 
length at most $h$. 
We employ the \emph{junction tree} approach, discussed in \Cref{sec:intro}. 
We formally define an $\lc$-length-constrained junction tree.

\begin{definition}
  Given an instance $(G, D, \lc)$ of LC-DSF, an 
  \emph{$\lc$-length-constrained junction tree} that covers terminal pairs 
  $D_H \subseteq D$ 
  is a subgraph $H \subseteq G$ with a root $v \in V(H)$ such that 
  for each $(s_i, t_i) \in D_H$, $H$ contains an $s_i$-$v$ path and a $v$-$t_i$ path, 
  each of length at most $\lc$. The \emph{density} of $H$ is 
  $c(H)/|D_H|$. 
\end{definition}

We first show, in \Cref{sec:jt_existence}, that given any instance to planar LC-DSF, 
there exists a low-density length-constrained 
junction tree. In \Cref{sec:dsf_algo}, we describe the approximation algorithm to find such 
a junction tree, and conclude the proof of \Cref{thm:lcdsf_main}. 

\subsection{Existence of a good junction tree}
\label{sec:jt_existence}

We prove the following theorem. 

\begin{theorem}
\label{thm:lcdsf_existence-main}
  Given an instance $(G, D, \lc)$ of length-constrained planar DSF, there exists 
  a $(3\lc)$-length-constrained junction tree of density 
  at most $O(\log^2 \nt) \cdot \opt/\nt$. 
\end{theorem}

For planar-DSF \emph{without} length-constraints, \cite{dsf_soda} showed that 
every instance of planar-DSF contains a low density junction tree. We outline the approach 
of \cite{dsf_soda}, as we follow a similar framework. The proof proceeds by considering 
an optimum solution $E^* \subseteq E$ and finding several junction trees in $E^*$ that are 
essentially disjoint and cover a large fraction of terminal pairs. The construction of 
these junction trees follows three steps. The first is a ``layering'' of a directed graph,
where each layer has a tree-like structure, and all paths in the original graph are contained 
in at most two consecutive layers. This layering scheme was given by Thorup \cite{Thorup04}
and also used by Kawarabayashi and Sidiropoulos \cite{KawarabayashiS21} 
to obtain improved upper bounds on the multicommodity flow-cut gap in directed planar graphs.
The second step is a recursive approach based on the planar separator theorem 
(\Cref{lem:undir_sep}). The third step is to show that in one level of recursion, 
if many $s_i$-$t_i$ paths pass through the separator, then there exist vertices on the 
separator that behave as roots of low-density junction trees. 

The proof of \Cref{thm:lcdsf_existence-main} follows the same three-step strategy.
First, we reduce to the case where the optimal solution is an \emph{$\lc$-length-constrained
$3$-layered digraph} (\Cref{lem:lcdsf_layers_main}).
\begin{definition}
  An $h$-length-constrained $3$-layered spanning tree $T$ of a digraph $G$ 
  is a spanning tree of the \emph{undirected version} of $G$, 
  such that any root-to-leaf path in $T$ is the concatenation of at 
  most $3$ dipaths in $G$, each of length at most $h$. 
  An $h$-length-constrained $3$-layered digraph is one that has an 
  $h$-length-constrained $3$-layered spanning tree. 
  The root of the digraph is the root of the spanning tree.
\end{definition}
The argument uses a similar decomposition as in \cite{dsf_soda},
but requires additional care to handle lengths.
Second, we apply the planar separator and corresponding divide-and-conquer 
argument to reduce to a so-called ``one-path'' setting
(\Cref{lem:lcdsf_separator_main}).
This step is essentially identical to the corresponding argument in
\cite{dsf_soda}.
The third step is a new contribution: handling the one-path setting under
length constraints. The argument of
\cite{dsf_soda} no longer applies, and we give a new argument
(\Cref{lem:lcdsf_one_path_main}).

\subsubsection{Reduction to length-constrained layered digraphs}
\label{sec:layers}

First, we reduce to the case where the optimal solution is an 
$h$-length-constrained $3$-layered digraph, by proving the following lemma. 

\begin{lemma}
\label{lem:lcdsf_layers_main}
  Let $(G, D, \lc)$ be an instance of length-constrained planar DSF and let 
  $G^* = (V^*, E^*)$ be an optimal solution. Then, 
  there exists a subgraph $H \subseteq G^*$ and a subset of terminals 
  $D_H \subseteq D$ such that the following holds.
  \begin{enumerate}
    \item $H$ is an $\lc$-length-constrained $3$-layered digraph with root $v$;
    \item For each $(s_i, t_i) \in D_H$, there exists an $s_i$-$t_i$ path 
    in $H \setminus \{v\}$ of length at most $\lc$;
    \item $H$ is a minor of $G^*$, and $H \setminus \{v\}$ is a subgraph of $G^*$. 
    \item $\frac{c(E(H))}{|D_H|} \leq 3 \cdot \frac{c(E^*)}{|D|}$. 
  \end{enumerate}
\end{lemma}

To prove \Cref{lem:lcdsf_layers_main}, we will decompose $G^*$ into a 
series of $3$-layered digraphs and show that at least one satisfies 
the necessary properties. We assume $G^*$ is weakly connected; else, 
this decomposition can be applied separately on each weakly connected component. 
Let $v_0$ be an arbitrary node in $V^*$. Let $L_0$ be the 
set of all nodes in $V^*$ that are reachable from $v_0$ (in $G^*$)
with a path of length at most $h$. We then define alternating layers
until all vertices have been covered by a layer. 

\[
  L_j = \begin{cases}
    \{u \in V^* \setminus \cup_{j' < j} L_{j'}: \ell(u, L_{j-1}) \leq \lc\} &j \text{ is odd} \\
    \{u \in V^* \setminus \cup_{j' < j} L_{j'}: \ell(L_{j-1}, u) \leq \lc\} &j \text{ is even} \\
  \end{cases}
\]

Let $q$ denote the total number of layers. For each layer $j \leq q-2$, we define $G_j$ 
to be the graph obtained from $G^*$ by contracting $\cup_{j' < j} L_{j'}$ into a new 
node $v_j$, and deleting $\cup_{j'' > j+2} L_{j''}$. 

\begin{claim}
\label{claim:lcdsf_3layers_path}
  For each demand pair $(s_i, t_i) \in D$, there exists $j \in \{0, \dots, q-2\}$ such that 
  $G_j \setminus \{v_j\}$ contains an $s_i$-$t_i$ path of length at most $\lc$. 
\end{claim}
\begin{proof}
  Let $P_i$ be an $s_i$-$t_i$ path of length at most $\lc$ in $G^*$; such a path must 
  exist as $G^*$ is a feasible solution to the given LC-DSF instance. Let $j$ be the 
  minimum index such that $P_i \cap L_j \neq \emptyset$. We assume $j$ is odd; the 
  case where $j$ is even is similar. Let $u$ be in $P_i \cap L_j$. Then, since 
  $\ell(u, t_i) \leq \lc$, $t_i$ must be contained in $\cup_{j' \leq j+1} L_{j'}$.
  Since all of $P_i$ can reach $t_i$ with a path of length at most $\lc$, 
  $P_i \subseteq \cup_{j' \leq j+2} L_{j'}$.
  Therefore, since $L_j$ is the first layer that intersects with $P_i$, $P_i \subseteq 
  L_j \cup L_{j+1} \cup L_{j+2} \subseteq G_j \setminus \{v_j\}$.\footnote{We 
  point out a minor technicality -- this claim 
  is not true when $j = 0$. However, one can circumvent this issue by creating a 
  new dummy vertex as the root of $G_0$ and adding a $0$-length edge to $v_0$.} 
\end{proof}

\begin{remark}
  When considering planar-DSF \emph{without} length constraints,
  \cite{dsf_soda} shows that all $s_i$-$t_i$ paths are contained in at most 2 layers. 
  With length constraints, however, $3$ layers are needed; see 
  \Cref{fig:3-layer-eg} 
  for an example of a graph in which the $s$-$t$ path spans 3 layers. 
\end{remark}

\begin{figure}
  \centering
  \includegraphics[width=0.6\linewidth]{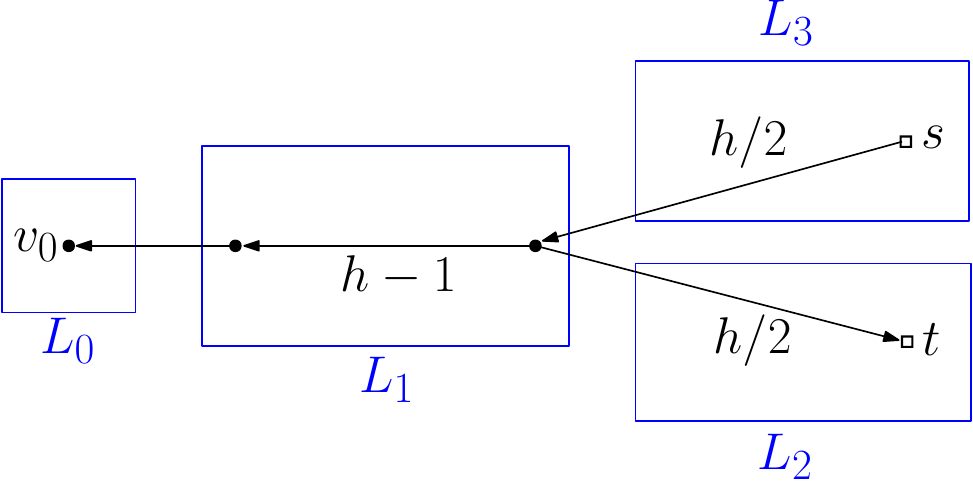}
  \caption{Example of graph in which $s$-$t$ path spans 3 layers, 
  $L_1, L_2, L_3$.}
  \label{fig:3-layer-eg}
\end{figure}

\begin{proof}[Proof of \Cref{lem:lcdsf_layers_main}]
  For each $j \in \{0, \dots, q-2\}$, let $D_j$ denote the set of terminal pairs 
  for which $G_j \setminus v_j$ contains an $s_i$-$t_i$ path of length at 
  most $\lc$. 
  We show that the first three properties required by \Cref{lem:lcdsf_layers_main} 
  hold for all $(G_j, D_j)$. The first 
  and third properties follow by construction: 
  each $G_j$ is an $\lc$-length-constrained $3$-layered digraph with root $v_j$, 
  and $G_j$ is constructed from $G^*$ via contraction (to form the root) and deletion. 
  The second property follows by definition of $D_j$. 
  
  To conclude, we show that there exists some $(G_j, D_j)$ that satisfies the fourth 
  property. By \Cref{claim:lcdsf_3layers_path}, $\cup_{j} D_j = D$. 
  Furthermore, it is not difficult to 
  see that $\sum_{j < q-1} c(G_j) \leq 3c(E^*)$, since each edge of $E^*$ 
  is either contained in one layer $L_j$, or goes between two layers $L_j$ and $L_{j+1}$ 
  -- in either case, it exists only in $G_j$, 
  $G_{j-1}$, and $G_{j-2}$. By an averaging argument, there must be some 
  $j \in [q-2]$ such that 
  \[\frac {c(G_j)}{|D_j|} \leq \frac{\sum_{j} c(G_j)}{\sum_{j} |D_j|}
  \leq \frac {3c(E^*)}{|D|}. \qedhere\]
\end{proof}

\subsubsection{Divide-and-conquer via planar separator}
\label{sec:dsf_separator}

Next, we use the fact that planar digraphs have good separators 
to reduce to the so-called ``one-path'' case. We prove the following 
lemma. 

\begin{lemma}
\label{lem:lcdsf_separator_main}
  Let $G^* = (V^*, E^*)$ be a planar $h$-length-constrained $3$-layered 
  digraph with root $v$, 
  and let $D \subseteq V^* \times V^*$ be a 
  set of demand pairs such that for each 
  $(s_i, t_i) \in D$, $G^* \setminus \{v\}$ has an $s_i$-$t_i$ path 
  of length at most $\lc$. Then, there exists a subgraph $H \subseteq G^* \setminus \{v\}$, 
  an $\lc$-length path $P \subseteq H$, and subset of terminals 
  $D_H \subseteq D$ satisfying the following.
  \begin{enumerate}
    \item For each $(s_i, t_i) \in D_H$, there exists an $\lc$-length 
    $s_i$-$t_i$ path in $H$ that intersects $P$;
    \item $\frac{c(E(H))}{|D_H|} \leq O(\log |D|) \frac{c(E^*)}{|D|}$. 
  \end{enumerate}
\end{lemma}

The proof follows the same argument as that of \cite{dsf_soda}; 
we rewrite it here for completeness and borrow the terminology. 
Fix an instance $(G, D)$ of planar-DSF and a feasible solution $E^*$ satisfying 
the conditions outlined in the statement of \Cref{lem:lcdsf_separator_main} above. Let 
$T^* \subseteq E^*$ be an 
$h$-length-constrained $3$-layered spanning tree of $G^*$. We will follow a
recursive process to partition $D$ into subsets on which we build junction trees.

\begin{itemize}
  \item Set vertex weights as $w(s_i) = w(t_i) = 1$ for all 
  $(s_i, t_i) \in D$, and $0$ otherwise.
  \item Consider the undirected version of $E^*$ and apply \Cref{lem:undir_sep}
  on the undirected version of spanning tree $T^*$ with vertex weights $w$. 
  From this, we obtain $u_1, u_2, u_3$.
  \item Since $T^*$ is an $h$-length-constrained $3$-layered tree, 
  each (undirected) path $P_{T^*}(v, u_i)$ consists of at most $3$ 
  dipaths, each of length at most $h$. Remove the root $v$ and let 
  $Q_i^j$ for $i, j \leq 3$ denote the $3$ dipaths of 
  $P_{T^*}(v, u_i) \setminus \{v\}$. Let $S_0 = \cup_{i, j \leq 3} Q_i^j$ 
  denote this collection of at most $9$ dipaths, we call this the 
  \emph{separator}.
  \item Define $D_0 \subseteq D$ to be the set
  of all terminal pairs $(s_i, t_i)$ such that $E^*$ contains an $h$-length 
  $s_i$-$t_i$ path that intersects
  one of the dipaths in $S_0$. Equivalently, $(s_i, t_i) \in D_0$ iff
  there exists an $h$-length $s_i$-$t_i$ path $P_i \subseteq E^*$ such that
  $V(P_i) \cap V(S_0) \neq \emptyset$.
  We say these terminal pairs are \emph{separated.}
  \item Let $\calC_1$ be the set of weakly connected
  components of $G \setminus (\cup_{i \in [3]} P_{T^*}(v, u_i))$; we 
  drop ``weakly connected'' and simply refer to these
  as ``components'' in the remainder of this section. Note that
  each $C \in \calC_1$ has at most half the total number of terminals.
\end{itemize}

We then follow a recursive process, similar to \algsep\ described 
in \Cref{sec:dst}. For each $C \in \calC_1$, 
we contract $S_0$ into $v$ 
and recurse on the sub-instance consisting of 
$G^*[C \cup \{v\}]$ with terminal set 
$\{(s_i, t_i) \in D: s_i, t_i \in C\} \setminus D_0$. 
It is easy to verify that $G^*[C \cup \{v\}]$ is an $h$-length-constrained 
$3$-layered digraph 
with root $v$. Furthermore, each $(s_i, t_i)$ pair in the new 
terminal set has an $h$-length-bounded path that does not intersect $v$; 
else, $(s_i, t_i)$ would have been included in $D_0$. 
Therefore, we can apply the same process as above, and repeat 
this recursive process until each component has at most one terminal.
Since the number of terminals halves
at each step, the recursion depth is at most $\lceil \log 2\nt
\rceil = \lceil \log \nt \rceil + 1$.

We use the following notation. For an instance defined on component 
$C$ at recursive depth $j$, we let $S_j^C, D_j^C$ denote the separator 
and corresponding separated terminal pairs. We let 
$\calC_j$ denote the set of all components at recursive depth $j$ 
(in particular, $\calC_0 = \{G^*$\}), and denote by $S_j := \cup_{C \in \calC_{j}} S_j^C$
and $D_j := \cup_{C \in \calC_j} D_j^C$. The following claim is simple.

\begin{claim}
\label{claim:all-terminals-sep}
  $D \subseteq \cup_{j = 0}^{\lceil \log \nt \rceil + 1} D_j$.
\end{claim}
\begin{proof}
  Fix $(s_i, t_i) \in D$ and let $P_i$ be an $h$-length $s_i$-$t_i$ path in $E^* \setminus \{v\}$. 
  Let $j$ be
  the first recursive level such that $P_i$ intersects $S_j$; such a level
  must exist since by the last step of recursion, $s_i$ and $t_i$ are in different
  components. Let $S_j^C \subseteq S_j$ be the first 
  separator that $P_i$ intersects. Then $P_i \subseteq C$; else $P_i$ would have 
  intersected an earlier separator. Thus $(s_i, t_i) \in D_j^C \subseteq D_j$.
\end{proof}

\begin{corollary}
\label{cor:good-recursion-level}
  There exists a level of recursion $j^* \in \{0, \dots, \lceil \log \nt \rceil + 1\}$
  such that $|D_{j^*}| \geq \frac \nt {\lceil \log \nt \rceil + 2}$.
\end{corollary}

\begin{proof}[Proof of \Cref{lem:lcdsf_separator_main}]
  Let $j^*$ be the recursion level given by
  \Cref{cor:good-recursion-level} such that
  $|D_{j^*}| \geq \frac \nt {\lceil \log \nt \rceil + 2}$.
  Since components of $\calC_{j^*}$ are disjoint, 
  $\sum_{C \in \calC_{j^*}} c(E(C)) \leq c(E^*)$ and $\{D_{j^*}^C\}_{C \in \calC_{j^*}}$ 
  form a partition of $D_{j^*}$. Using an averaging argument, there exists 
  a component $C$ such that $c(E(C))/|D_{j^*}^C| \leq c(E^*)/|D_{j^*}|$. 
  Furthermore, since $S_{j^*}^C$ consists of at most $9$ dipaths, there exists 
  at least one dipath $Q^*$ such that at least $\frac 1 9$ terminal pairs of 
  $D_{j^*}^C$ intersect $Q^*$; we denote this set of terminals by $D^*$. 

  We claim $C, Q^*, D^*$ satisfy the required properties of the lemma statement. 
  The first follows by construction; for all terminal pairs $(s_i, t_i) \in D^*$,
  there exists an $h$-length $s_i$-$t_i$ path in $C$ that intersects $Q^*$. 
  For the second, 
  $
    \frac{c(E(C))}{|D^*|} \leq 9 \cdot \frac{c(E(C))}{|D_{j^*}^C|}
    \leq 9 \cdot \frac{c(E^*)}{|D_{j^*}|} 
    \leq O(\log k) \frac{c(E^*)}{|D|}.
  $
\end{proof}

\subsubsection{One Path Case}
\label{sec:one_path}

To conclude, we show that in the ``one-path'' setting, there always exists a
low-density junction tree. 
This subsection contains the core technical contribution over 
\cite{dsf_soda}, since previously used techniques for handling the one-path 
setting do not extend to length-constraints. 

\begin{lemma}
\label{lem:lcdsf_one_path_main}
  Let $G^* = (V^*, E^*)$ be a planar subgraph that contains an $h$-length 
  path $P \subseteq G^*$, and let $D \subseteq V^* \times V^*$ 
  be a set of demand pairs such that for each $(s_i, t_i) \in D$,
  there exists an $h$-length $s_i$-$t_i$ path in $G^*$ that intersects $P$. 
  Then, $G^*$ contains a $3\lc$-length-constrained junction tree of density 
  $O(\log \nt) c(E^*)/\nt$, where $\nt = |D|$. 
\end{lemma}

Let $v_1, \dots, v_{|P|}$ denote the vertices of $P$. 
For each terminal pair $(s_i, t_i) \in D$, the lemma assumption 
guarantees an $\lc$-length $s_i$-$t_i$ path in $G^*$ that intersects 
$P$. Let $a_i'$ denote the first vertex of $P$ on this path, 
and let $b_i'$ denote the last vertex of $P$ on this path; 
note that $a_i' \leq_P b_i'$. Let $P_{s_i}$ be the subpath 
from $s_i$ to $a_i'$, and let $P_{t_i}$ be the subpath from 
$b_i'$ to $t_i$; both have length at most $\lc$.
We can assume without loss of generality that $E^* = 
P \cup (\cup_{i \in [\nt]} (P_{s_i} \cup P_{t_i}))$, since this 
satisfies the lemma assumptions. 
The main challenge with length constraints in the one-path setting 
is in handling edges of $E^* \setminus P$. In particular, 
the argument in \cite{dsf_soda} relies on the fact that 
for any $i, i' \in [\nt]$, if
$P_{s_i} \cap P_{s_{i'}} \neq \emptyset$, then $a_i' = a_{i'}'$.
This no longer holds, as 
demonstrated by \Cref{fig:lc-one-path}.

\begin{figure}
  \centering
  \includegraphics[width=0.4\linewidth]{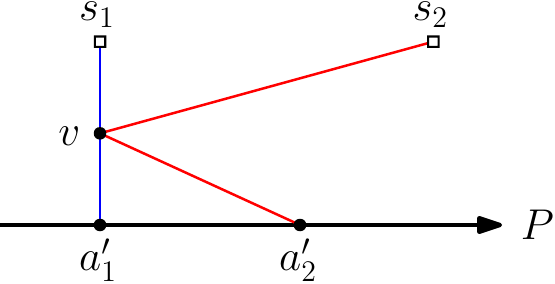}
  \caption{Suppose $\ell(s_1, v) = 1, \ell(v, a_1') = \lc-1, 
  \ell(s_2, v) = \lc-1, \ell(v, a_2') = 1$. The argument in 
  \cite{dsf_soda} would re-route $s_2$ to 
  $P$ via the path from $v$ to $a_1'$; however, this would increase 
  the length of the path from $s_2$ to $P$ to $2\lc-2$. In general, 
  merging arbitrarily many terminals would make these paths to 
  $P$ arbitrarily long.}
  \label{fig:lc-one-path}
\end{figure}

We adopt a different perspective. For each $(s_i, t_i) \in D$, 
let $a_i$
denote the first node in $P$ that $s_i$ can reach within 
a path of length \emph{at most $2\lc$}, and let 
$b_i$ denote the last node in $P$ that can reach $t_i$
within length at most $2\lc$. 
We use the following simple claim.
\begin{claim}
  \label{claim:lcdsf_path_order}
  For all $i \in [\nt]$, $a_i \leq_P a_i' \leq_P b_i' \leq_P b_i$.
\end{claim}
\begin{proof}
  By definition, $a_i' \leq_P b_i'$. For the other two inequalities,
  $a_i \leq_P a_i'$, since $a_i$ is the \emph{first} node on 
  $P$ that $s_i$ can reach within length $2\lc$, and $P_{s_i}$ is an 
  $s$-$P$ path of length $\lc$. Similarly, $b_i' \leq_P b_i$, 
  since $b_i$ is the 
  \emph{last} node on $P$ that can reach $t_i$ within length $2\lc$, and $P_{t_i}$ 
  is a path from $P$ to $t_i$ of length $\lc$. 
\end{proof}

Let $\mu$ be the midpoint of 
the path $P$ (i.e., $\mu = v_{\lfloor|P|/2\rfloor}$). 
We partition $D$ into three sets: 
\[
\begin{cases}
L = \{(s_i, t_i): a_i \leq_P b_i <_P \mu\}\\
R = \{(s_i, t_i): b_i \geq_P a_i  >_P \mu\} \\
M = \{(s_i, t_i): a_i \leq_P \mu \leq_P b_i\}
\end{cases}
\]

\begin{claim}
\label{claim:lcdsf_onepath_junction}
  $E^*$ contains a $3\lc$-length-constrained junction tree 
  with root $\mu$ on terminals $M$.  
\end{claim}
\begin{proof}
  It suffices to show that for each $(s_i, t_i) \in M$, there exist $3\lc$-length 
  paths in $E^*$ from $s_i$ to $\mu$ and from $\mu$ to $t_i$. By definition, 
  $\ell(s_i, a_i) \leq 2\lc$. Furthermore, since $a_i \leq_P \mu$ and 
  $P$ is an $\lc$-length path, $\ell(a_i, \mu) \leq \lc$. Therefore, 
  $\ell(s_i, \mu) \leq 3\lc$. The argument for $t_i$ is similar. 
\end{proof}

We follow a recursive process. First, we construct the junction tree $H_\mu$ 
with root $\mu$ given by \Cref{claim:lcdsf_onepath_junction}. 
We conflate indices with vertices, and let $(\mu-1)$ and $(\mu+1)$ denote the 
vertices directly before and after $\mu$ on $P$. At a high level, the 
idea is then to recurse on the sub-paths $P[0, \mu-1]$, $P[\mu+1, v_{|P|}]$, 
with corresponding paths in $E \setminus P$ to connect terminals. 
Formally, let $E_L = P[0, \mu-1] \cup \left(\cup_{(s_i, t_i) \in L} (P_{s_i} \cup P_{t_i})\right)$,
and $E_R = P[\mu+1, v_{|P|}] \cup \left(\cup_{(s_i, t_i) \in R} (P_{s_i} \cup P_{t_i})\right)$.
We recurse on the edge-induced subgraph
$G_L := G[E_L]$ with terminal set $L$, and 
$G_R := G[E_R]$ with terminal set $R$.
The following claim follows directly from \Cref{claim:lcdsf_path_order},
and proves that these are valid recursive calls.

\begin{claim}
\label{claim:lcdsf_onepath_recursion}
  For each $(s_i, t_i) \in D$, recall that $a_i', b_i'$ denote the endpoints 
  on $P$ of $P_{s_i}, P_{t_i}$ respectively. If $(s_i, t_i) \in L$, then 
  $a_i', b_i' \in P[0, \mu-1]$, and if $(s_i, t_i) \in R$, then 
  $a_i', b_i' \in P[\mu+1, v_{|P|}]$.
\end{claim}

We continue 
this recursive process until $P$ consists of one vertex $v$, in which case 
all remaining terminal pairs can be routed through $v$. This process 
terminates after $O(\log|P|) = O(\log \nt)$ recursive steps, since $|P|$ 
is halved at each step. The following key lemma bounds the total 
cost of the junction trees constructed.
Let $E_s = E(\cup_{i \in [\nt]} P_{s_i})$
and $E_t = E(\cup_{i \in [\nt]} P_{t_i})$. 

\begin{claim}
\label{claim:lcdsf_onepath_cost}
  $E_s \cap E_L$ is disjoint from $E_s \cap E_R$. Similarly, 
  $E_t \cap E_L$ is disjoint from $E_t \cap E_R$. 
\end{claim}
\begin{proof}
  We prove the claim for $E_s$; the argument for $E_t$ is similar. 
  Suppose for the sake of contradiction there exists an edge $e = (u,v)$ in 
  $E_s \cap E_L \cap E_R$. Since $e \in E_L \cap E_s$, there must be 
  a terminal $s_i$ such that $e \in P_{s_i}$ and $a_i \leq_P b_i <_P \mu$. By 
  \Cref{claim:lcdsf_path_order}, $a_i' <_P \mu$. Similarly, 
  since $e \in E_R$, there exists $s_{j}$ such that $e \in P_{s_j}$ 
  and $a_j >_P \mu$. Consider the path 
  $P' = P_{s_j}[s_j, u] \circ P_{s_i}[u, a_i']$. Since $P_{s_j}$ and 
  $P_{s_i}$ are both $\lc$-length paths, $P'$ has length at most $2\lc$. 
  Thus, $s_j$ can reach $a_i'$ with a path of length $\leq 2\lc$, and 
  $a_i' <_P a_j$, a contradiction. 
\end{proof}

\begin{proof}[Proof of \Cref{lem:lcdsf_one_path_main}]
  The recursive process described yields a series of junction trees 
  denoted $H_v$ for some vertices $v \in P$ that, in total, cover all terminal pairs. 
  By \Cref{claim:lcdsf_onepath_junction}, 
  each $H_v$ is a $3\lc$-length-constrained junction tree. 
  To bound the cost, it is helpful to think of $E_s, E_t$ and $E(P)$ as disjoint 
  sets; this can be achieved by making at most $3$ copies of each edge in $E^*$. 
  Then, by \Cref{claim:lcdsf_onepath_cost}, $E_L$ is disjoint from $E_R$. 
  Therefore, each edge is in at most one junction tree per recursive level, 
  so it is in at most $O(\log \nt)$ total junction trees. Therefore, 
  the total cost of all junction trees is at most 
  $3 \cdot O(\log \nt) c(E^*) = O(\log \nt) c(E^*)$. By an averaging argument, 
  there must be one junction tree of density at most $O(\log \nt) c(E^*)/\nt$. 
\end{proof}

To conclude, 
the proof of \Cref{thm:lcdsf_existence-main} follows from 
\Cref{lem:lcdsf_layers_main}, \Cref{lem:lcdsf_separator_main}, and 
\Cref{lem:lcdsf_one_path_main}.

\subsection{Finding an Approximate Low-Density Junction Tree}
\label{sec:dsf_algo}

The main lemma we prove in this section is the following. 
\begin{lemma}
\label{lem:lcdsf_algo-main}
  Suppose there exists a bicriteria $(\alpha, \beta)$-approximation for
  LC-DST in planar graphs
  with respect to the optimal solution to \ref{LC-DST-LP}. Then, given a
  planar-LC-DSF instance $(G, D, \lc)$, there exists an efficient algorithm to
  obtain a $\beta \lc$-length-constrained junction tree 
  of $G$ of density at most $O(\alpha \cdot \log \nt)$
  times the optimal.
\end{lemma}

The algorithm is as follows. First, we guess the root of the junction tree; 
we can run this algorithm for each candidate root $v \in V$ and return the 
junction tree with minimum density. Next, we solve the following LP relaxation 
for the min-density junction tree rooted at $v$. 
We follow a similar structure to that of \ref{LC-DST-LP},
with additional variables $y_{i}$ for each
$i \in [\nt]$ to indicate whether or not the pair $(s_i, t_i)$ is 
covered in the solution. For each $(s_i, t_i) \in D$, 
we let $\calP_{s_i}^{\lc}$ denote the set of all $\lc$-length-constrained 
paths from $s_i$ to $v$, and $\calP_{t_i}^{\lc}$ denote the set of 
all $\lc$-length-constrained paths from $v$ to $t_i$. 
Following \ref{LC-DST-LP}, we have variables $x_e$ for each $e \in E$ indicating 
whether an edge $e$ is in junction tree, and 
variables $f_p$ indicating which paths are used to connect $s_i$ to $v$ and 
$v$ to $t_i$. 
The resulting density of the junction tree would be $(\sum_{e \in E} c(e)x_e)/
(\sum_{i \in [\nt]} y_{i})$; we normalize $\sum_{i \in [\nt]} y_{i} = 1$.

\begin{equation}
\label{LC-Den-LP}
\tag{LC-Den-LP}
\begin{aligned}
  \min\quad \sum_{e\in E}c(e)x_e& \\
  s.t.\quad \sum_{p \in \calP_{s_i}^{\lc}} f_p &\geq y_i \qquad~
  \forall i \in [\nt] \\
  \sum_{p \in \calP_{t_i}^{\lc}} f_p &\geq y_i \qquad~
  \forall i \in [\nt] \\
  \sum_{p \in \calP_{s_i}^{\lc}, e \in p} f_p &\leq x_e \qquad 
  \forall e \in E, i \in [\nt] \\
  \sum_{p \in \calP_{t_i}^{\lc}, e \in p} f_p &\leq x_e \qquad 
  \forall e \in E, i \in [\nt] \\
  \sum_{i \in [\nt]} y_{i} = 1 \\
  x_e, f_p, y_i &\geq 0 \qquad~
  \forall e\in E, p \in \cup_{i \in [\nt]} \left(\calP_{s_i}^\lc \cup \calP_{t_i}^\lc\right), i \in [k]
\end{aligned}
\end{equation}

Let $x^*, f^*, y^*$ be an optimal solution to \ref{LC-Den-LP}. 
Partition the set of 
terminal pairs into groups $D_0, \dots, D_{\log \nt}$, where 
$(s_i, t_i) \in D_j$ if 
$y_{i}^* \in (\frac 1 {2^{j+1}}, \frac 1 {2^j}]$.
We will abuse notation and sometimes write $i \in D_j$ for ease.

\begin{claim}
\label{claim:junction_good_group}
	There exists $\theta \in \{0, \dots, \log \nt\}$ such that 
	$\sum_{i \in D_{\theta}} y_i^* \geq 1/(2(\log \nt+1))$. 
\end{claim}
\begin{proof}
  If $(s_i, t_i) \in D \setminus (\cup_{j =0}^{\log \nt} D_j)$, then
  $y_{i}^* \leq \frac 1 {2^{\log \nt+1}} = \frac 1 {2\nt}$. Thus, 
  $\sum_{i \notin \cup_{j = 0}^{\log \nt} D_j} y_{i}^*
  \leq \sum_{i \notin \cup_{j = 0}^{\log \nt} D_j} \frac 1 {2\nt} \leq \frac 1 2$.
  Recall that $\sum_{i \in [\nt]} y_{i}^* = 1$. Therefore, 
  $\sum_{i \in \cup_{j = 0}^{\log \nt} D_j} y_{i}^* \geq \frac 1 2$.
  Since there are $\log \nt + 1$ disjoint groups, there must be at least one group
  $D_\theta$ with $\sum_{i \in D_{\theta}} y_i^*$ at least $1/(2(\log \nt+1))$.
\end{proof}

Let $\theta$ be given by \Cref{claim:junction_good_group}. We use the
$(\alpha, \beta)$-approximation algorithm for LC-DST twice: first, we consider the instance
on $G$ with root $v$ and terminal set $D_\theta^t = \{t_i: (s_i, t_i) \in D_\theta\}$, 
and obtain a solution $T_t$. Second, we let $G^R$ be obtained from $G$ 
by reversing the direction of all edges. We consider the LC-DST instance on 
$G^R$ with root $v$ and terminal set 
$D_\theta^s = \{s_i: (s_i, t_i) \in D_\theta\}$ and obtain a solution
$T_s'$ in $G^R$. Note that $T_s'$ contains an $h$-length-constrained $v$-$s_i$ path 
in $G^R$ for all $s_i \in D_\theta^s$. 
Thus, if we reverse all edges of $T_s'$, the resulting graph $T_s$ is a subgraph of $G$
containing an $h$-length-constrained $s_i$-$v$ path for all $s_i \in D_\theta^s$. 
Therefore, $T = T_t \cup T_s$ is a $(\beta h)$-length-constrained junction tree 
on terminal pairs $D_\theta$. 

\begin{claim}
\label{claim:junction_lp_feasible}
Let $\lambda = 2^{\theta + 1}$. Then $(\lambda x^*, \lambda f^*)$ is a feasible 
solution to \ref{LC-DST-LP} on both of the following instances:
  \begin{itemize}
    \item $G$ with terminal set $D_\theta^t$,
    \item $G^R$ with terminal set $D_\theta^s$.
  \end{itemize}
  In both cases, $f^*$ is restricted to relevant paths for terminals in $D_\theta^t$, 
  $D_\theta^s$. 
\end{claim}
\begin{proof}
  We prove the claim for $G^R$ with terminal set $D_\theta^s$, as the
  argument for $(G, D_\theta^t)$ is analogous. 
  Recall that $\calP_{s_i}^h$ is the set of all $h$-length-constrained 
  $s_i$-$v$ paths in $G$. Thus, each path in $\calP_{s_i}^h$ is an $h$-length 
  constrained $v$-$s_i$ path in $G^R$. 
  Since $(x^*, f^*, y^*)$ is a feasible solution to \ref{LC-Den-LP}, 
  $\sum_{p \in \calP_{s_i}^{\lc}} f_p^* \geq y_i^*$. 
  Furthermore, since $i \in D_\theta$, $y_i^* > \frac 1 {2^{\theta+1}}$. 
  Thus,
  $\sum_{p \in \calP_{s_i}^{\lc}} \lambda f_p^* \geq \lambda y_i^* > 1.$

  The fact that $(\lambda x^*, \lambda f^*)$ satisfies the second constraint of 
  \ref{LC-DST-LP} follows immediately 
  from the fact that $(x^*, f^*)$ satisfies the third constraint of \ref{LC-DST-LP},
  since both $x^*$ and $f^*$ are scaled by the same value.
\end{proof}

\begin{claim}
\label{claim:junction_density}
  The density of $T$ is at most $O(\alpha \log \nt) \sum_{e \in E} c(e) x_e^*$.
\end{claim}
\begin{proof}
  The algorithms used to construct $T_s$ and $T_t$ are $(\alpha, \beta)$-approximations 
  with respect to the LP. This, along with \Cref{claim:junction_lp_feasible}, 
  implies that the costs of $T_t$ and $T_s$ are each upper bounded by
  $\alpha \lambda \sum_{e \in E} c(e) x_e^*$, where $\lambda = 2^{\theta + 1}$. 

  To bound $|D_{\theta}|$,
  observe that $\sum_{i \in D_\theta} y_{i}^* \leq \sum_{i \in D_{\theta}} 1/2^\theta
  = |D_\theta| /2^{\theta}$. Furthermore, recall that $\theta$ was defined by 
  \Cref{claim:junction_good_group} such that $\sum_{i \in D_{\theta}} y_{i}^* 
  \geq 1/(2(\log k + 1))$. Therefore, 
  $|D_{\theta}| \geq 2^{\theta} \sum_{i \in D_\theta} y_{i}^* \geq 
  \frac {2^\theta}{2(\log k + 1)}$. The density of $T$ is at most

  \begin{align*}
    &\frac{c(T_t) + c(T_s)}{|D_{\theta}|}
  \leq \frac {2\log \nt + 2} {2^{\theta}} \left(2\alpha \lambda \sum_{e \in E} c(e)x_e^*\right)
  = 8\alpha(\log \nt + 1) \sum_{e \in E} c(e) x_e^*.
  \end{align*}
\end{proof}

\Cref{lem:lcdsf_algo-main} follows from \Cref{claim:junction_density}
and the fact that $T$ is a $(\beta h)$-length-constrained junction tree.
Combining \Cref{lem:lcdsf_algo-main} with the bicriteria LP-competitive approximation
for LC-DST in planar digraphs (\Cref{thm:lcdst_frac})
then yields \Cref{thm:lcdsf_algo-main}.

\begin{theorem}
\label{thm:lcdsf_algo-main}
  Given an instance $(G, D, \lc)$ of LC-DSF in planar digraphs, there exists an 
  efficient algorithm to obtain an $O(\log \nt)$-length-constrained 
  junction tree of $G$ of density at most 
  $O(\log^3 \nt)$ times the optimal.
\end{theorem}

This, along with \Cref{thm:lcdsf_existence-main}, 
yields an efficient algorithm to find an 
$O(\log \nt)$-length-constrained junction tree of density at most
$O(\log^5 \nt) \opt/\nt$. By the standard covering argument, this yields 
a bicriteria $(O(\log^6 \nt), O(\log \nt))$-approximation for 
LC-DSF in planar digraphs, concluding the proof of 
\Cref{thm:lcdsf_main}.

\bibliographystyle{plainurl}
\bibliography{refs}

@article{warburton1987approximation,
  title={Approximation of Pareto optima in multiple-objective, shortest-path problems},
  author={Warburton, Arthur},
  journal={Operations research},
  volume={35},
  number={1},
  pages={70--79},
  year={1987},
  publisher={INFORMS}
}

@article{chekuri2026node,
  title={Node-Weighted Multicut in Planar Digraphs},
  author={Chekuri, Chandra and Jain, Rhea},
  journal={arXiv preprint arXiv:2601.20038},
  year={2026}
}

@inproceedings{grigorescu2025directed,
  title={Directed Buy-At-Bulk Spanners},
  author={Grigorescu, Elena and Kumar, Nithish and Lin, Young-San},
  booktitle={Approximation, Randomization, and Combinatorial Optimization. Algorithms and Techniques (APPROX/RANDOM 2025)},
  pages={22--1},
  year={2025},
  organization={Schloss Dagstuhl--Leibniz-Zentrum f{\"u}r Informatik}
}

@article{hassin1992approximation,
  title={Approximation schemes for the restricted shortest path problem},
  author={Hassin, Refael},
  journal={Mathematics of Operations research},
  volume={17},
  number={1},
  pages={36--42},
  year={1992},
  publisher={INFORMS}
}

@inproceedings{bernstein2012near,
  title={Near Linear Time $(1+\epsilon)$-Approximation for Restricted Shortest Paths in Undirected Graphs},
  author={Bernstein, Aaron},
  booktitle={Proceedings of the twenty-third annual ACM-SIAM symposium on Discrete Algorithms},
  pages={189--201},
  year={2012},
  organization={SIAM}
}

@article{joksch1966shortest,
  title={The shortest route problem with constraints},
  author={Joksch, Hans C},
  journal={Journal of Mathematical analysis and applications},
  volume={14},
  number={2},
  pages={191--197},
  year={1966},
  publisher={Elsevier}
}

@book{johnson1979computers,
  title={Computers and intractability: A guide to the theory of NP-completeness},
  author={Johnson, David S and Garey, Michael R},
  year={1979},
  publisher={WH Freeman}
}

@article{lorenz2001simple,
  title={A simple efficient approximation scheme for the restricted shortest path problem},
  author={Lorenz, Dean H and Raz, Danny},
  journal={Operations Research Letters},
  volume={28},
  number={5},
  pages={213--219},
  year={2001},
  publisher={Elsevier}
}

@inproceedings{goel2001efficient,
  title={Efficient computation of delay-sensitive routes from one source to all destinations},
  author={Goel, Ashish and Ramakrishnan, Kajamalai G and Kataria, Deepak and Logothetis, Dimitris},
  booktitle={Proceedings IEEE INFOCOM},
  volume={2},
  pages={854--858},
  year={2001},
  organization={IEEE}
}

@article{abboud2024reachability,
  title={Reachability preservers: New extremal bounds and approximation algorithms},
  author={Abboud, Amir and Bodwin, Greg},
  journal={SIAM Journal on Computing},
  volume={53},
  number={2},
  pages={221--246},
  year={2024},
  publisher={SIAM}
}

@article{laekhanukit2024integrality,
  title={On the Integrality Gap of Directed {Steiner} Tree {LPs} with Relatively Integral Solutions},
  author={Laekhanukit, Bundit},
  journal={arXiv preprint arXiv:2412.10744},
  year={2024}
}

@article{hershkowitz2025planar,
  title={Planar Length-Constrained Minimum Spanning Trees},
  author={Hershkowitz, D Ellis and Huang, Richard Z},
  journal={arXiv preprint arXiv:2510.09002},
  year={2025},
  note={To appear in STOC 2026}
}

@article{GhugeN22,
  title={Quasi-polynomial algorithms for submodular tree orienteering and directed network design problems},
  author={Ghuge, Rohan and Nagarajan, Viswanath},
  journal={Mathematics of Operations Research},
  volume={47},
  number={2},
  pages={1612--1630},
  year={2022},
  publisher={INFORMS}
}

@inproceedings{ChekuriP05,
  title={A recursive greedy algorithm for walks in directed graphs},
  author={Chekuri, Chandra and Pal, Martin},
  booktitle={46th annual IEEE symposium on foundations of computer science (FOCS'05)},
  pages={245--253},
  year={2005},
  organization={IEEE}
}

@inproceedings{HalperinK03,
  title={Polylogarithmic inapproximability},
  author={Halperin, Eran and Krauthgamer, Robert},
  booktitle={Proceedings of the thirty-fifth annual ACM symposium on Theory of computing},
  pages={585--594},
  year={2003}
}

@InProceedings{hop_esa,
  author =	{Chekuri, Chandra and Jain, Rhea},
  title =	{Approximation Algorithms for Hop Constrained and Buy-At-Bulk Network Design via Hop Constrained Oblivious Routing},
  booktitle =	{32nd Annual European Symposium on Algorithms (ESA 2024)},
  pages =	{41:1--41:21},
  series =	{Leibniz International Proceedings in Informatics (LIPIcs)},
  year =	{2024},
  volume =	{308},
  editor =	{Chan, Timothy and Fischer, Johannes and Iacono, John and Herman, Grzegorz},
  publisher =	{Schloss Dagstuhl -- Leibniz-Zentrum f{\"u}r Informatik},
  address =	{Dagstuhl, Germany},
  doi =		{10.4230/LIPIcs.ESA.2024.41}
}

@InProceedings{polymatroid_esa,
  author =	{Chekuri, Chandra and Jain, Rhea and Kulkarni, Shubhang and Zheng, Da Wei and Zhu, Weihao},
  title =	{{From Directed Steiner Tree to Directed Polymatroid Steiner Tree in Planar Graphs}},
  booktitle =	{32nd Annual European Symposium on Algorithms (ESA 2024)},
  pages =	{42:1--42:19},
  series =	{Leibniz International Proceedings in Informatics (LIPIcs)},
  ISBN =	{978-3-95977-338-6},
  ISSN =	{1868-8969},
  year =	{2024},
  volume =	{308},
  editor =	{Chan, Timothy and Fischer, Johannes and Iacono, John and Herman, Grzegorz},
  publisher =	{Schloss Dagstuhl -- Leibniz-Zentrum f{\"u}r Informatik},
  address =	{Dagstuhl, Germany},
  URL =		{https://drops.dagstuhl.de/entities/document/10.4230/LIPIcs.ESA.2024.42},
  URN =		{urn:nbn:de:0030-drops-211134},
  doi =		{10.4230/LIPIcs.ESA.2024.42},
}

@inproceedings{dsf_soda,
author = {Chandra Chekuri and Rhea Jain},
title = {A Polylogarithmic Approximation for {Directed Steiner Forest} in Planar Digraphs},
booktitle = {Proceedings of the 2025 Annual ACM-SIAM Symposium on Discrete Algorithms (SODA)},
pages = {2095-2110},
year = {2025},
doi = {10.1137/1.9781611978322.67}
}

@inproceedings{bab_stoc,
  author       = "Chandra Chekuri and Rhea Jain",
  title        = "A Polylogarithmic Approximation for Buy-at-Bulk Network Design with Protection", 
  booktitle    = "Proceedings of the Fifty-Eighth Annual ACM Symposium on Theory of Computing (STOC)",
  year         = "2026",
  note         = "to be published"
}

@article{chks09,
author = {Chekuri, C. and Hajiaghayi, M. T. and Kortsarz, G. and Salavatipour, M. R.},
title = {Approximation Algorithms for Nonuniform Buy-at-Bulk Network Design},
journal = {SIAM Journal on Computing},
volume = {39},
number = {5},
pages = {1772-1798},
year = {2010},
doi = {10.1137/090750317}
}

@inproceedings{hop_congestion21,
author = {Ghaffari, Mohsen and Haeupler, Bernhard and Zuzic, Goran},
title = {Hop-Constrained Oblivious Routing},
year = {2021},
isbn = {9781450380539},
publisher = {Association for Computing Machinery},
address = {New York, NY, USA},
url = {https://doi.org/10.1145/3406325.3451098},
doi = {10.1145/3406325.3451098},
booktitle = {Proceedings of the 53rd Annual ACM SIGACT Symposium on Theory of Computing},
pages = {1208-1220},
numpages = {13},
location = {Virtual, Italy},
series = {STOC 2021}
}

@inproceedings{hop_distance21,
author = {Haeupler, Bernhard and Hershkowitz, D. Ellis and Zuzic, Goran},
title = {Tree Embeddings for Hop-Constrained Network Design},
year = {2021},
isbn = {9781450380539},
publisher = {Association for Computing Machinery},
address = {New York, NY, USA},
url = {https://doi.org/10.1145/3406325.3451053},
doi = {10.1145/3406325.3451053},
booktitle = {Proceedings of the 53rd Annual ACM SIGACT Symposium on Theory of Computing},
pages = {356-369},
numpages = {14},
location = {Virtual, Italy},
series = {STOC 2021}
}

@INPROCEEDINGS {filtser22,
author = {A. Filtser},
booktitle = {2021 IEEE 62nd Annual Symposium on Foundations of Computer Science (FOCS)},
title = {Hop-Constrained Metric Embeddings and their Applications},
year = {2022},
volume = {},
issn = {},
pages = {492-503},
doi = {10.1109/FOCS52979.2021.00056},
url = {https://doi.ieeecomputersociety.org/10.1109/FOCS52979.2021.00056},
publisher = {IEEE Computer Society},
address = {Los Alamitos, CA, USA},
month = feb
}

@article{acsz11,
author = {Antonakopoulos, Spyridon and Chekuri, Chandra and Shepherd, Bruce and Zhang, Lisa},
title = {Buy-at-Bulk Network Design with Protection},
journal = {Mathematics of Operations Research},
volume = {36},
number = {1},
pages = {71-87},
year = {2011},
doi = {10.1287/moor.1110.0484}
}

@article{kortsarz_nutov11,
title = {Approximating some network design problems with node costs},
journal = {Theoretical Computer Science},
volume = {412},
number = {35},
pages = {4482-4492},
year = {2011},
issn = {0304-3975},
doi = {https://doi.org/10.1016/j.tcs.2011.04.013},
url = {https://www.sciencedirect.com/science/article/pii/S0304397511003021},
author = {Guy Kortsarz and Zeev Nutov},
}

@article{akgun_tansel11,
  title={New formulations of the Hop-Constrained Minimum Spanning Tree problem via {Miller-Tucker-Zemlin} constraints},
  author={Akg{\"u}n, Ibrahim and Tansel, Barbaros {\c{C}}},
  journal={European Journal of Operational Research},
  volume={212},
  number={2},
  pages={263--276},
  year={2011},
  publisher={Elsevier}
}

@article{ba92,
author = {Balakrishnan, Anantaram and Altinkemer, Kemal},
title = {Using a Hop-Constrained Model to Generate Alternative Communication Network Design},
journal = {ORSA Journal on Computing},
volume = {4},
number = {2},
pages = {192-205},
year = {1992},
doi = {10.1287/ijoc.4.2.192}
}

@article{pirkul_soni03,
title = {New formulations and solution procedures for the hop constrained network design problem},
journal = {European Journal of Operational Research},
volume = {148},
number = {1},
pages = {126-140},
year = {2003},
issn = {0377-2217},
doi = {https://doi.org/10.1016/S0377-2217(02)00366-1},
url = {https://www.sciencedirect.com/science/article/pii/S0377221702003661},
author = {Hasan Pirkul and Samit Soni},
}

@article{dahl98,
title = {The 2-hop spanning tree problem},
journal = {Operations Research Letters},
volume = {23},
number = {1},
pages = {21-26},
year = {1998},
doi = {https://doi.org/10.1016/S0167-6377(98)00029-7},
author = {Geir Dahl}
}

@article{althaus05,
title = {Approximating k-hop minimum-spanning trees},
journal = {Operations Research Letters},
volume = {33},
number = {2},
pages = {115-120},
year = {2005},
issn = {0167-6377},
doi = {https://doi.org/10.1016/j.orl.2004.05.005},
author = {Ernst Althaus and Stefan Funke and Sariel Har-Peled and Jochen Könemann and Edgar A. Ramos and Martin Skutella}}

@article{kls05,
  title={Approximating the degree-bounded minimum diameter spanning tree problem},
  author={K{\"o}nemann, Jochen and Levin, Asaf and Sinha, Amitabh},
  journal={Algorithmica},
  volume={41},
  pages={117--129},
  year={2005},
  publisher={Springer}
}

@inproceedings{kp97,
author = {Kortsarz, Guy and Peleg, David},
title = {Approximating Shallow-Light Trees},
year = {1997},
isbn = {0898713900},
publisher = {Society for Industrial and Applied Mathematics},
address = {USA},
booktitle = {Proceedings of the Eighth Annual ACM-SIAM Symposium on Discrete Algorithms},
pages = {103–110},
numpages = {8},
location = {New Orleans, Louisiana, USA},
series = {SODA '97}
}

@InProceedings{kp06,
author="Kantor, Erez
and Peleg, David",
editor="Calamoneri, Tiziana
and Finocchi, Irene
and Italiano, Giuseppe F.",
title="Approximate Hierarchical Facility Location and Applications to the Shallow Steiner Tree and Range Assignment Problems",
booktitle="Algorithms and Complexity",
year="2006",
publisher="Springer Berlin Heidelberg",
address="Berlin, Heidelberg",
pages="211--222",
isbn="978-3-540-34378-3"
}

@INPROCEEDINGS{ravi94,
  author={Ravi, R.},
  booktitle={Proceedings 35th Annual Symposium on Foundations of Computer Science}, 
  title={Rapid rumor ramification: approximating the minimum broadcast time}, 
  year={1994},
  volume={},
  number={},
  pages={202-213},
  doi={10.1109/SFCS.1994.365693}}

@article{marathe98,
title = {Bicriteria Network Design Problems},
journal = {Journal of Algorithms},
volume = {28},
number = {1},
pages = {142-171},
year = {1998},
issn = {0196-6774},
doi = {https://doi.org/10.1006/jagm.1998.0930},
url = {https://www.sciencedirect.com/science/article/pii/S0196677498909300},
author = {Madhav V Marathe and R Ravi and Ravi Sundaram and S.S Ravi and Daniel J Rosenkrantz and Harry B Hunt},
}

@article{hks09,
  title={Approximating buy-at-bulk and shallow-light $k$-Steiner trees},
  author={Hajiaghayi, Mohammad Taghi and Kortsarz, Guy and Salavatipour, Mohammad R},
  journal={Algorithmica},
  volume={53},
  pages={89--103},
  year={2009},
  publisher={Springer}
}

@article{ks11,
author = {Khani, M. Reza and Salavatipour, Mohammad R.},
title = {Improved Approximations for Buy-at-Bulk and Shallow-Light $k$-Steiner Trees and $(k,2)$-Subgraph},
year = {2016},
issue_date = {February  2016},
publisher = {Springer-Verlag},
address = {Berlin, Heidelberg},
volume = {31},
number = {2},
issn = {1382-6905},
url = {https://doi.org/10.1007/s10878-014-9774-5},
doi = {10.1007/s10878-014-9774-5},
journal = {J. Comb. Optim.},
month = feb,
pages = {669–685},
numpages = {17}
}

@article{dkr16,
author = {Dinitz, Michael and Kortsarz, Guy and Raz, Ran},
title = {Label Cover Instances with Large Girth and the Hardness of Approximating Basic K-Spanner},
year = {2016},
issue_date = {February 2016},
publisher = {Association for Computing Machinery},
address = {New York, NY, USA},
volume = {12},
number = {2},
issn = {1549-6325},
url = {https://doi.org/10.1145/2818375},
doi = {10.1145/2818375},
journal = {ACM Trans. Algorithms},
month = dec,
articleno = {25},
numpages = {16}
}

@inproceedings{HZ25_sosa_mst,
author = {D Ellis Hershkowitz and Richard Z Huang},
title = {Simple Length-Constrained Minimum Spanning Trees},
booktitle = {2025 Symposium on Simplicity in Algorithms (SOSA)},
chapter = {},
pages = {341-349},
year = {2025},
doi = {10.1137/1.9781611978315.25},
URL = {https://epubs.siam.org/doi/abs/10.1137/1.9781611978315.25},
eprint = {https://epubs.siam.org/doi/pdf/10.1137/1.9781611978315.25}
}

@article{FriggstadM23,
  title={{A $O(\log k)$-Approximation for Directed Steiner Tree in Planar Graphs}},
  author={Friggstad, Zachary and Mousavi, Ramin},
  journal={ACM Transactions on Algorithms},
  volume={21},
  number={4},
  pages={1--14},
  year={2025},
  publisher={ACM New York, NY},
  note = "Preliminary version in ICALP 2023"
}

@article{Thorup04,
author = {Thorup, Mikkel},
title = {Compact oracles for reachability and approximate distances in planar digraphs},
year = {2004},
issue_date = {November 2004},
publisher = {Association for Computing Machinery},
address = {New York, NY, USA},
volume = {51},
number = {6},
issn = {0004-5411},
url = {https://doi.org/10.1145/1039488.1039493},
doi = {10.1145/1039488.1039493},
journal = {J. ACM},
month = nov,
pages = {993-1024},
numpages = {32}
}

@article{Charikaretal99,
  title={Approximation algorithms for directed {Steiner} problems},
  author={Charikar, Moses and Chekuri, Chandra and Cheung, To-Yat and Dai, Zuo and Goel, Ashish and Guha, Sudipto and Li, Ming},
  journal={Journal of Algorithms},
  volume={33},
  number={1},
  pages={73--91},
  year={1999},
  publisher={Elsevier}
}

@inproceedings{ZosinK02,
  title={On directed {Steiner} trees},
  author={Zosin, Leonid and Khuller, Samir},
  booktitle={Proceedings of ACM-SIAM SODA},
  pages={59--63},
  year={2002}
}

@article{LiL22,
  title={Polynomial integrality gap of flow {LP} for directed {Steiner} tree},
  author={Li, Shi and Laekhanukit, Bundit},
  journal={ACM Transactions on Algorithms},
  volume={21},
  number={1},
  pages={1--9},
  year={2024},
  publisher={ACM New York, NY}
}

@article{GrandoniLL22,
  title={${O}(\log^2 k/\log \log k)$-Approximation Algorithm for Directed {Steiner} Tree: A Tight Quasi-Polynomial Time Algorithm},
  author={Grandoni, Fabrizio and Laekhanukit, Bundit and Li, Shi},
  journal={SIAM Journal on Computing},
  volume={52},
  number={2},
  pages={298--322},
  year={2022},
  publisher={SIAM},
  note = "Preliminary version in Proc. of STOC 2019"
}

@article{Zelikovsky97,
  title={A series of approximation algorithms for the acyclic directed {Steiner} tree problem},
  author={Zelikovsky, Alexander},
  journal={Algorithmica},
  volume={18},
  number={1},
  pages={99--110},
  year={1997},
  publisher={Springer}
}

@article{DemaineHK14,
  title={Node-weighted {Steiner} tree and group {Steiner} tree in planar graphs},
  author={Demaine, Erik D and Hajiaghayi, MohammadTaghi and Klein, Philip N},
  journal={ACM Transactions on Algorithms (TALG)},
  volume={10},
  number={3},
  pages={1--20},
  year={2014},
  publisher={ACM New York, NY, USA}
}

@article{BorradaileKM09,
  title={An ${O}(n \log n)$ approximation scheme for {Steiner} tree in planar graphs},
  author={Borradaile, Glencora and Klein, Philip and Mathieu, Claire},
  journal={ACM Transactions on Algorithms (TALG)},
  volume={5},
  number={3},
  pages={1--31},
  year={2009},
  publisher={ACM New York, NY, USA}
}

@inproceedings{FriggstadKKLST14,
  title={Linear programming hierarchies suffice for directed {Steiner} tree},
  author={Friggstad, Zachary and K{\"o}nemann, Jochen and Kun-Ko, Young and Louis, Anand and Shadravan, Mohammad and Tulsiani, Madhur},
  booktitle={International Conference on Integer Programming and Combinatorial Optimization},
  pages={285--296},
  year={2014},
  organization={Springer}
}

@article{Rothvoss11,
  author       = {Thomas Rothvo{\ss}},
  title        = {Directed {Steiner} Tree and the {Lasserre} Hierarchy},
  journal      = {CoRR},
  volume       = {abs/1111.5473},
  year         = {2011},
  url          = {http://arxiv.org/abs/1111.5473},
  eprinttype    = {arXiv},
  eprint       = {1111.5473}
}

@article{ChekuriEGS11,
  title={Set connectivity problems in undirected graphs and the directed steiner network problem},
  author={Chekuri, Chandra and Even, Guy and Gupta, Anupam and Segev, Danny},
  journal={ACM Transactions on Algorithms (TALG)},
  volume={7},
  number={2},
  pages={1--17},
  year={2011},
  publisher={ACM New York, NY, USA}
}

@inproceedings{KawarabayashiS21,
  author       = {Ken{-}ichi Kawarabayashi and Anastasios Sidiropoulos},
  title        = {Embeddings of Planar Quasimetrics into Directed $\ell_1$ and Polylogarithmic Approximation for Directed Sparsest-Cut},
  booktitle    = {62nd {IEEE} Annual Symposium on Foundations of Computer Science, {FOCS} 2021, Denver, CO, USA, February 7-10, 2022},
  pages        = {480--491},
  publisher    = {{IEEE}},
  year         = {2021},
  url          = {https://doi.org/10.1109/FOCS52979.2021.00055},
  doi          = {10.1109/FOCS52979.2021.00055}
}

@article{feldman_improved_2012,
	series = {{JCSS} {Knowledge} {Representation} and {Reasoning}},
	title = {Improved approximation algorithms for {Directed} {Steiner} {Forest}},
	volume = {78},
	issn = {0022-0000},
	url = {https://www.sciencedirect.com/science/article/pii/S0022000011000584},
	doi = {10.1016/j.jcss.2011.05.009},
	number = {1},
	urldate = {2024-06-23},
	journal = {Journal of Computer and System Sciences},
	author = {Feldman, Moran and Kortsarz, Guy and Nutov, Zeev},
	month = jan,
	year = {2012},
	pages = {279--292}
}

@inproceedings{cohen-addad_bypassing_2022,
	address = {New York, NY, USA},
	title = {Bypassing the surface embedding: approximation schemes for network design in minor-free graphs},
	shorttitle = {Bypassing the surface embedding},
	doi = {10.1145/3519935.3520049},
	booktitle = {Proceedings of the 54th {Annual} {ACM} {SIGACT} {Symposium} on {Theory} of {Computing}},
	publisher = {Association for Computing Machinery},
	author = {Cohen-Addad, Vincent},
	month = jun,
	year = {2022},
	pages = {343--356},
}

@InProceedings{FriggstadM21,
  author =	{Friggstad, Zachary and Mousavi, Ramin},
  title =	{A Constant-Factor Approximation for Quasi-Bipartite Directed {Steiner} Tree on Minor-Free Graphs},
  booktitle =	{Approximation, Randomization, and Combinatorial Optimization. Algorithms and Techniques (APPROX/RANDOM 2023)},
  pages =	{13:1--13:18},
  year =	{2023},
  URL =		{https://drops.dagstuhl.de/entities/document/10.4230/LIPIcs.APPROX/RANDOM.2023.13},
  URN =		{urn:nbn:de:0030-drops-188389},
  doi =		{10.4230/LIPIcs.APPROX/RANDOM.2023.13},
}

@article{BermanBKRY13,
title = {Approximation algorithms for spanner problems and Directed {Steiner} Forest},
journal = {Information and Computation},
volume = {222},
pages = {93-107},
year = {2013},
note = {38th International Colloquium on Automata, Languages and Programming (ICALP 2011)},
issn = {0890-5401},
doi = {https://doi.org/10.1016/j.ic.2012.10.007},
author = {Piotr Berman and Arnab Bhattacharyya and Konstantin Makarychev and Sofya Raskhodnikova and Grigory Yaroslavtsev},
}

@article{CDKL20,
  title={Approximating spanners and directed steiner forest: Upper and lower bounds},
  author={Chlamt{\'a}{\v{c}}, Eden and Dinitz, Michael and Kortsarz, Guy and Laekhanukit, Bundit},
  journal={ACM Transactions on Algorithms (TALG)},
  volume={16},
  number={3},
  pages={1--31},
  year={2020},
  publisher={ACM New York, NY, USA}
}

@article{LiptonTarjan79,
author = {Lipton, Richard J. and Tarjan, Robert Endre},
title = {A Separator Theorem for Planar Graphs},
journal = {SIAM Journal on Applied Mathematics},
volume = {36},
number = {2},
pages = {177-189},
year = {1979},
doi = {10.1137/0136016},
URL = {https://doi.org/10.1137/0136016},
eprint = {https://doi.org/10.1137/0136016}
}

@article{BateniHM11,
  title={Approximation schemes for Steiner forest on planar graphs and graphs of bounded treewidth},
  author={Bateni, M. and Hajiaghayi, M. and Marx, D.},
  journal={Journal of the ACM (JACM)},
  volume={58},
  number={5},
  pages={21},
  year={2011},
  publisher={ACM}
}

@book{my_thesis,
  title={Approximation Algorithms for Connectivity and Fault-Tolerant Network Design},
  author={Jain, Rhea},
  year={2026},
  publisher={University of Illinois at Urbana-Champaign}
}

@inproceedings{DodisK99,
  title={Design networks with bounded pairwise distance},
  author={Dodis, Yevgeniy and Khanna, Sanjeev},
  booktitle={Proceedings of the thirty-first annual ACM Symposium on Theory of Computing (STOC)},
  pages={750--759},
  year={1999},
  url={https://doi.org/10.1145/301250.301447}
}

\end{document}